\documentclass[10pt,journal,compsoc]{IEEEtran}
\usepackage{booktabs} 
\usepackage{multirow}
\usepackage{graphics}	
\usepackage{subfigure}
\pdfoutput=1
\graphicspath{{./figures/}}
\usepackage{tabularx}
\usepackage[english]{babel}
\usepackage[utf8]{inputenc}
\usepackage{amsfonts}
\usepackage{graphicx}
\usepackage{amsmath}
\usepackage{algpseudocode}
\usepackage[linesnumbered,ruled, vlined]{algorithm2e}
\usepackage{mathtools}
\usepackage{tikz}
\usepackage{url}
\usetikzlibrary{automata,positioning,arrows.meta}
\usepackage{xcolor,colortbl}

\newcommand{\kw}[1]{{\ensuremath {\mathsf{#1}}}\xspace}
\newcommand{\stitle}[1]{\vspace{0.75ex}\noindent{\bf #1}}
\newcommand{\stab}{\rule{0pt}{8pt}\\[-2ex]}

\newcommand{\eat}[1]{}

\usepackage{amsthm}
\newtheorem{theorem}{Theorem}
\newcommand{\ie}{\emph{i.e.,}\xspace}
\newcommand{\eg}{\emph{e.g.,}\xspace}
\newcommand{\wrt}{\emph{w.r.t.}\xspace}
\newcommand{\aka}{\emph{a.k.a.}\xspace}
\newcommand{\etc}{\emph{etc,}\xspace}
\newcommand{\rwm}{\kw{RWM}}

\newcommand{\ci}{\kw{Citeseer}}
\newcommand{\brainnet}{\kw{BrainNet}}
\newcommand{\mitdata}{\kw{RM}}
\newcommand{\sixng}{\kw{6-NG}}
\newcommand{\nineng}{\kw{9-NG}}
\newcommand{\dblp}{\kw{DBLP}}
\newcommand{\athle}{\kw{Athletics}}
\newcommand{\gene}{\kw{Genetic}}

\usepackage{pdfpages}

\begin{document}

\title{Random Walk on Multiple Networks}

\author{Dongsheng~Luo, Yuchen~Bian, Yaowei~Yan, Xiong~Yu, Jun~Huan, Xiao~Liu, Xiang~Zhang%
\IEEEcompsocitemizethanks{\IEEEcompsocthanksitem \vspace{-1ex}
D. Luo is with the Florida International University, Miami, FL, 33199.\protect\\
E-mail: dluo@fiu.edu.  \vspace{-1.5ex}}

\IEEEcompsocitemizethanks{\IEEEcompsocthanksitem
Y. Bian is with Amazon Search Science and AI, USA.\protect\\
E-mail: yuchbian@amazon.com. \vspace{-1.5ex}}

\IEEEcompsocitemizethanks{\IEEEcompsocthanksitem
Y. Yan is with Meta Platforms, Inc., USA.\protect\\
E-mail: yanyaw@meta.com. \vspace{-1.5ex}}

\IEEEcompsocitemizethanks{\IEEEcompsocthanksitem \vspace{-1ex}
X. Liu, and X. Zhang are with the Pennsylvania State University, State College, PA 16802.\protect\\
E-mail: \{xxl213, xzz89\}@psu.edu.  \vspace{-1.5ex}}

\IEEEcompsocitemizethanks{\IEEEcompsocthanksitem
X. Yu is with Case Western Reserve University.\protect\\
E-mail: xxy21@case.edu \vspace{-1.5ex}}
\IEEEcompsocitemizethanks{\IEEEcompsocthanksitem
J. Huan is with AWS AI Labs; work done before joining AWS.\protect\\
E-mail: lukehuan@amazon.com \vspace{-1.5ex}}

\thanks{}}

\markboth{Preprint}%
{\MakeLowercase{\textit{Dongsheng Luo et al.}}: Random Walk on Multi-Networks}

\IEEEtitleabstractindextext{%

\begin{abstract}
Random Walk is a basic algorithm to explore the structure of networks, which can be used in many tasks, such as local community detection and network embedding. Existing random walk methods are based on single networks that contain limited information. In contrast, real data often contain entities with different types or/and from different sources, which are comprehensive and can be better modeled by multiple networks. To take the advantage of rich information in multiple networks and make better inferences on entities, in this study, we propose random walk on multiple networks, \rwm. 
\rwm is flexible and supports both multiplex networks and general multiple networks, which may form many-to-many node mappings between networks. \rwm  sends a random walker on each network to obtain the local proximity (i.e., node visiting probabilities) w.r.t. the starting nodes. Walkers with similar visiting probabilities reinforce each other. We theoretically analyze the convergence properties of \rwm. Two approximation methods with theoretical performance guarantees are proposed for efficient computation. We apply \rwm in link prediction, network embedding, and local community detection. Comprehensive experiments conducted on both synthetic and real-world datasets demonstrate the effectiveness and efficiency of \rwm.   
\end{abstract}


}

\maketitle

\IEEEdisplaynontitleabstractindextext
\IEEEpeerreviewmaketitle

\section{Introduction}
Networks (graphs) are natural representations of relational data in the real world, with vertices modeling entities and edges describing relationships among entities. With the rapid growth of information, a large volume of network data is generated, such as social networks~\cite{leskovec2012learning}, author collaboration networks~\cite{newman2004coauthorship}, document citation networks~\cite{shuai2018query}, and biological networks~\cite{ni2018comclus}.
In many emerging applications, different types of vertices and relationships are obtained from different sources, which can be better modeled by multiple networks.

There are two main kinds of multiple networks. Fig.~\ref{fig:social_example} shows an example of multiple networks with the same node set and different types of edges (which are called \textit{multiplex networks}~\cite{MultilayerNetworks, MultiplexNet} or \textit{multi-layer networks}). Each node represents an employee of a university~\cite{kim2015community}. These three networks reflect different relationships between employees: co-workers, lunch-together, and Facebook-friends. Notice that more similar connections exist between two offline networks (i.e., co-worker and lunch-together) than that between offline and online relationships (i.e., Facebook-friends).

Fig.~\ref{fig:apnet} is another example from the DBLP dataset with multiple domains of nodes and edges (which is called \textit{multi-domain networks}~\cite{ni2018co,luo2020deep}). The left part is the author-collaboration network and the right part is the paper-citation network. A cross-edge connecting an author and a paper indicates that the author writes the paper. We see that authors and papers in the same research field may have dense cross-links. Multiplex networks are a special case of multi-domain networks with the same set of nodes and the cross-network relations are one-to-one mappings.

Given a starting node in a graph, we randomly move to one of its neighbors; then we move to a neighbor of the current node at random \etc. This process is called random walk on the graph~\cite{lovasz1993random}.
Random walk is an effective and widely used method to explore the local structures in networks, which serves as the basis of many advanced tasks such as network embedding~\cite{cao2015grarep, grover2016node2vec, perozzi2014deepwalk, tang2015line}, link prediction~\cite{martinez2017survey}, node/graph classification~\cite{klicpera2019diffusion,klicpera2018predict}, and graph generation~\cite{bojchevski2018netgan}. Various extensions, such as the random walk with restart (RWR)~\cite{tong2006fast}, second-order random walk~\cite{wu2017remember}, and multi-walker chain~\cite{bian2017many} have been proposed to analyze the network structure. Despite the encouraging progress,  most existing random walk methods focus on a single network merely. Some methods merge the multiple networks as a single network by summation and then apply traditional random walk\cite{de2015identifying}. However, these methods neglect types of edges and nodes and assume that all layers have the same effect, which may be too restrictive in most scenarios.

\begin{figure}[t!]
    \centering
    \subfigure[Co-worker]{\includegraphics[width=1.1in,trim=0cm 4cm 0cm 4cm, clip]{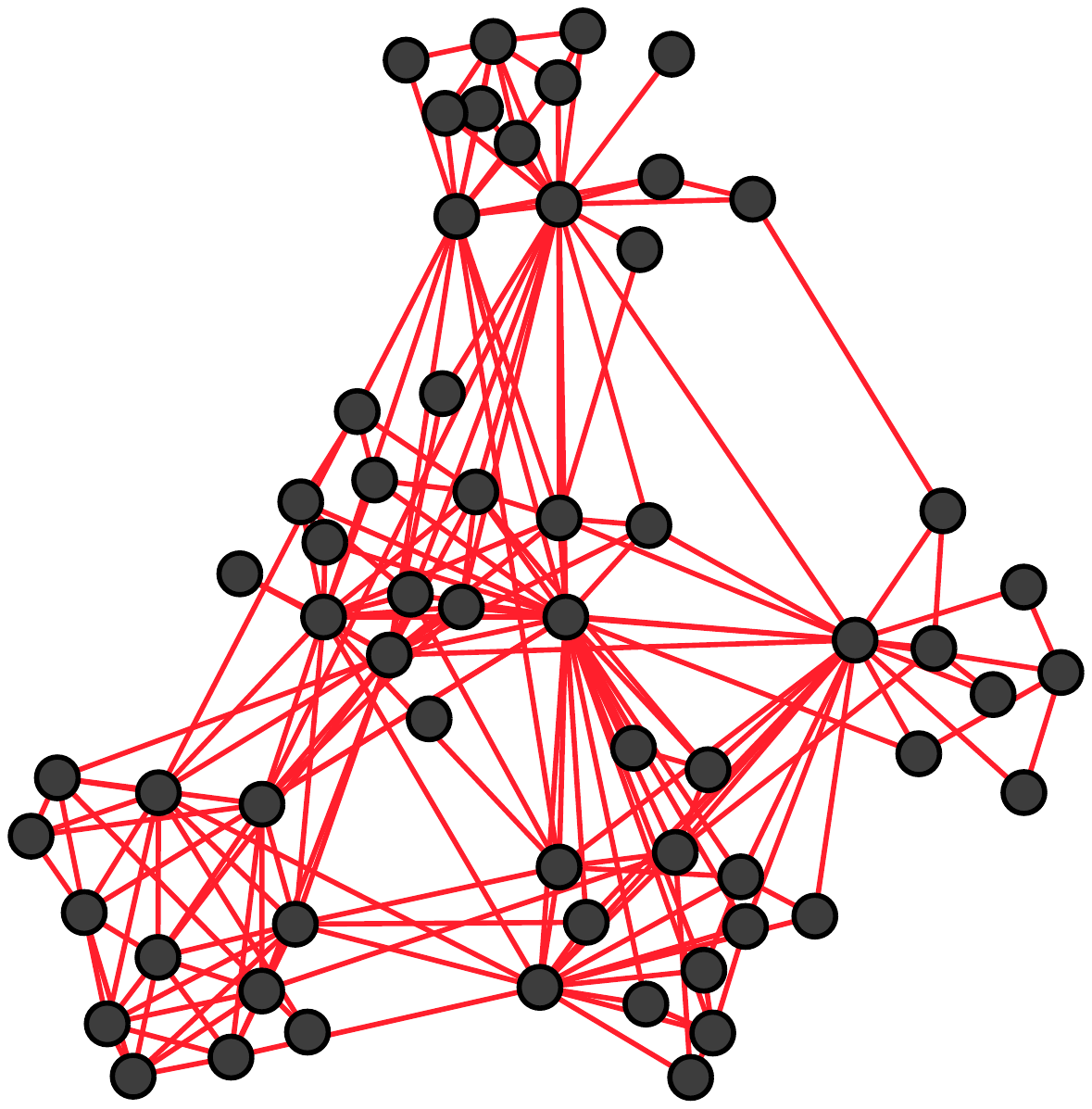}\label{fig:social_work}}\hspace{0.1cm}
    \subfigure[Lunch-together]{\includegraphics[width=1.1in,trim=0cm 4cm 0cm 4cm,clip]{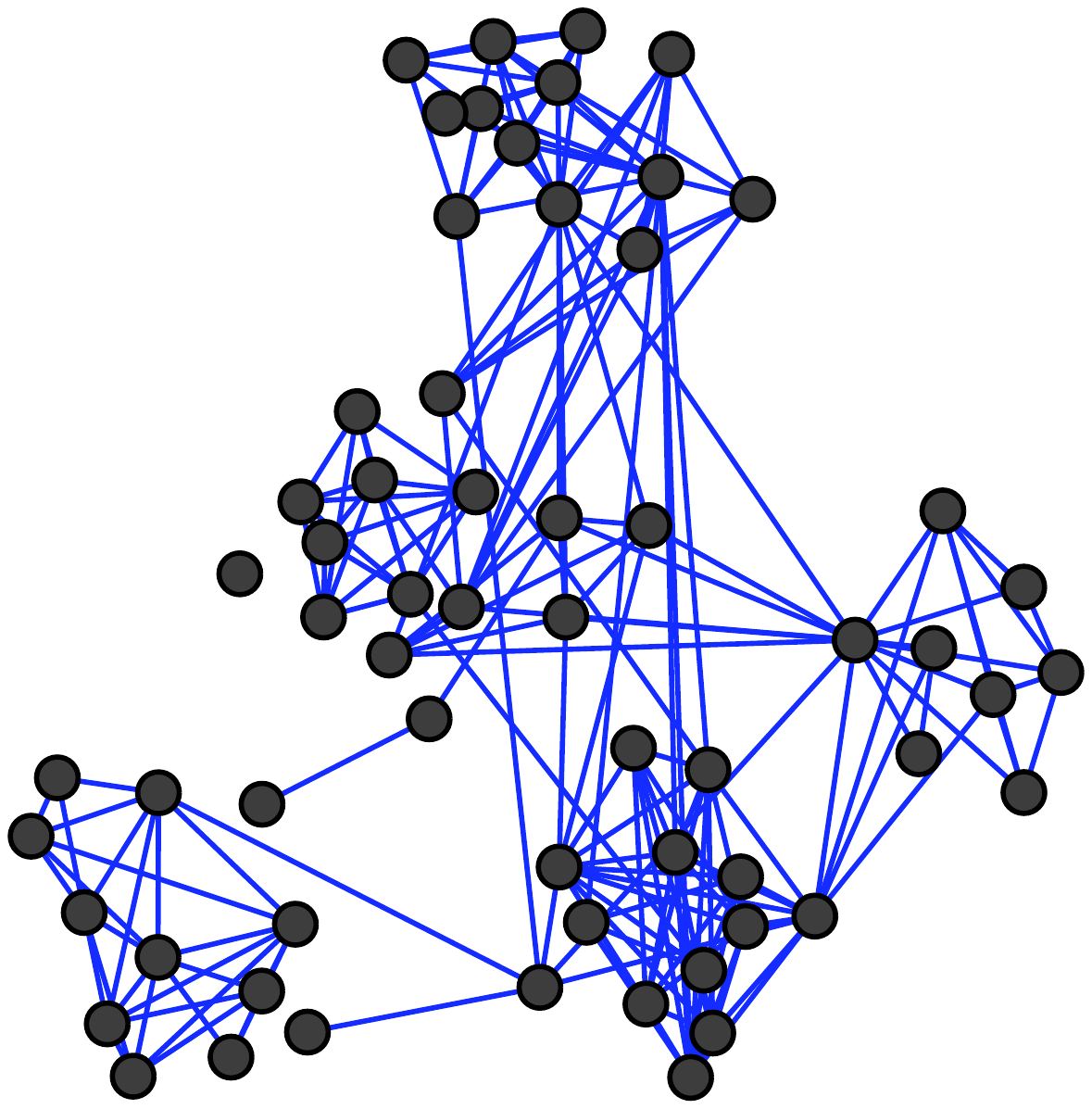}\label{fig:social_lunch}}\hspace{0.1cm}
    \subfigure[Facebook-friends]{\includegraphics[width=1.1in,trim=0cm 4cm 0cm 4cm, clip]{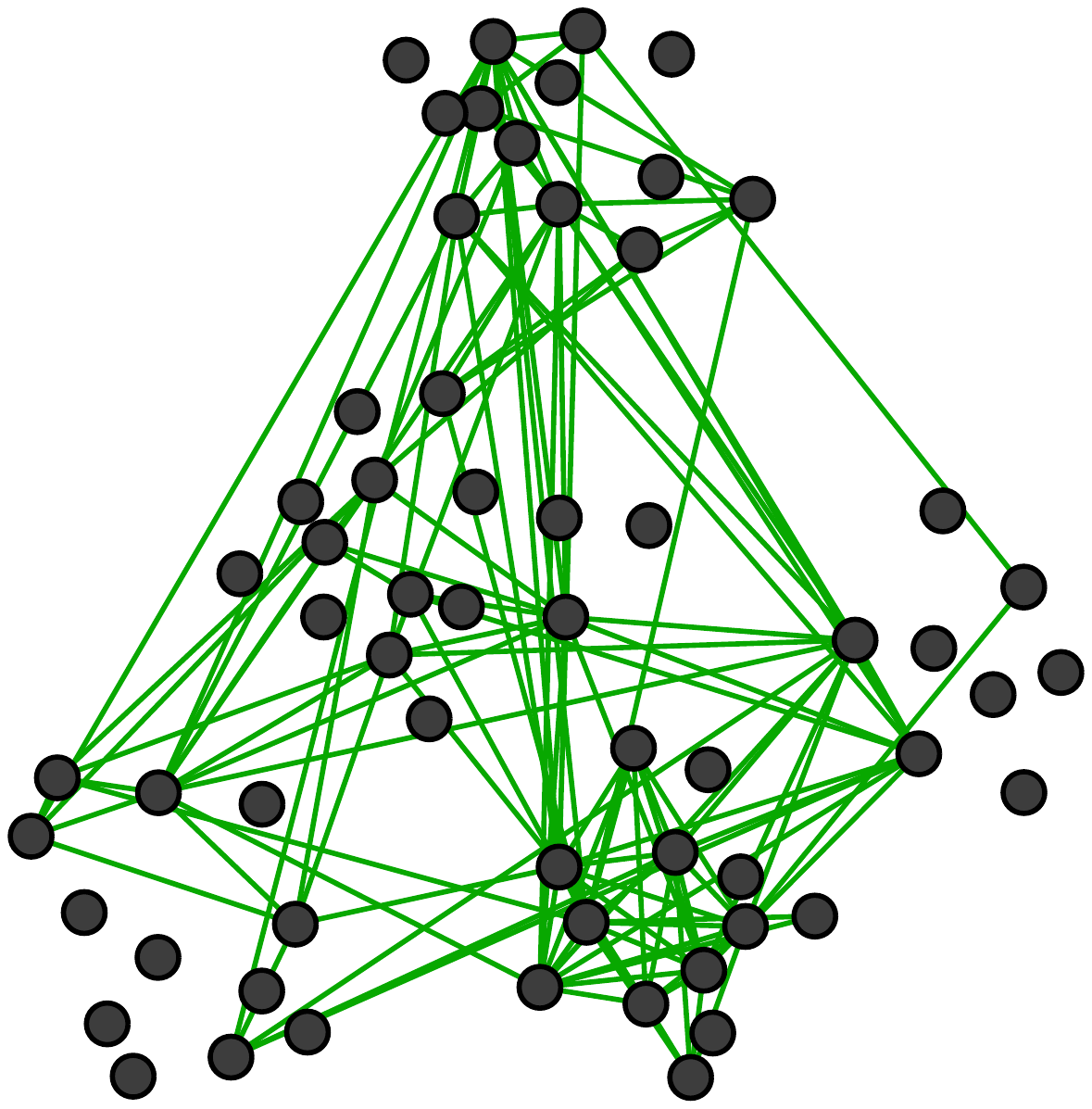}\label{fig:social_facebook}}
    \caption{An example of multiple social networks with the same node set. }
    \label{fig:social_example}
\end{figure}

\begin{figure}
    \centering
    \includegraphics[width=2.3in]{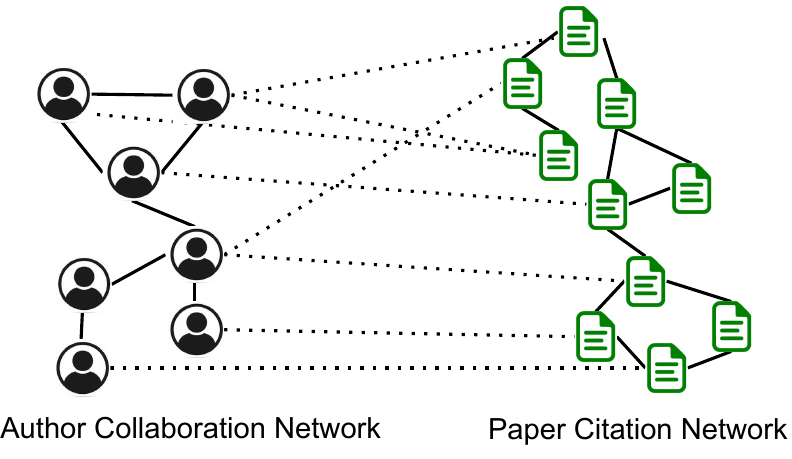}
    \caption{An example of general multiple networks. Many-to-many mappings exist between networks.}
    \label{fig:apnet}
\end{figure}

To this end, in this study, we propose a novel random walk method, \rwm (\underline{\textit{R}}andom \underline{\textit{W}}alk on \underline{\textit{M}}ultiple networks), to explore local structures in multiple networks.  The key idea is to integrate  complementary influences from multiple networks. 
We send out a random walker on each network to explore the network topology based on corresponding transition probabilities. 
For networks containing the starting node, the probability vectors of walkers are initiated with the starting node.  For other networks, probability vectors are initialized by transiting visiting probability via cross-connections among networks. Then probability vectors of the walkers are updated step by step. The vectors represent walkers' visiting histories. Intuitively, if two walkers share similar or relevant visiting histories (measured by cosine similarity and transited with cross-connections), it indicates that the visited local typologies in corresponding networks are relevant. And these two walkers should reinforce each other with highly weighted influences. On the other hand, if visited local typologies of two walkers are less relevant, smaller weights should be assigned. We update each walker's visiting probability vector by aggregating other walkers' influences at each step. In this way, the transition probabilities of each network are modified dynamically. Compared to traditional random walk models where the transition matrix is time-independent~\cite{andersen2006local,li2015uncovering,kloster2014heat,wu2015robust,wu2016remember}, \rwm can restrict most visiting probability in the local subgraph \wrt the starting node, in each network and ignore irrelevant or noisy parts. Theoretically, we provide rigorous analyses of the convergence properties of the proposed model. Two speeding-up strategies are developed as well. We conduct extensive experiments on real and synthetic datasets to demonstrate the advantages of our \rwm model on effectiveness and efficiency for local community detection, link prediction, and network embedding in multiple networks.

Our contributions are summarized as follows.
\begin{itemize}
    \item We propose a novel random walk model, \rwm, on multiple networks, which serves as a basic tool to analyze multiple networks. 
    \item Fast estimation strategies and sound theoretical foundations are provided to guarantee the effectiveness and efficiency of \rwm.
    \item Results of comprehensive experiments on synthetic and real multiple networks verify the advances of \rwm.
\end{itemize}

\section{Related Work}
\label{section:relatedwork}

\stitle{Random Walk.} In the single network random walk, a random walker explores the network according to the node-to-node transition probabilities, which are determined by the network topology. There is a stationary probability for visiting each node if the network is irreducible and aperiodic~\cite{langville2011google}. Random walk is the basic tool to measure the proximity of nodes in the network. Based on the random walk, various proximity measurement algorithms have been developed, among which PageRank~\cite{page1999pagerank}, RWR~\cite{tong2006fast}, SimRank~\cite{jeh2002simrank}, and MWC~\cite{bian2017many} have gained significant popularity.
In PageRank, instead of merely following the transition probability, the walker has a constant probability to jump to any node in the network at each time point. Random walk with restart, \aka personalized PageRank, is the query biased version of PageRank. At each time point, the walker has a constant probability to jump back to the query (starting) node. SimRank aims to measure the similarity between two nodes based on the principle that two nodes are similar if their neighbors are similar.  With one walker starting from each node,  the SimRank value between two nodes represents the expected number of steps required before two walkers meet at the same node if they walk in lock-step.  MWC sends out a group of walkers to explore the network. Each walker is pulled back by other walkers when deciding the next steps, which helps the walkers to stay as a group within the cluster. 

Very limited work has been done on the multiple network random walk. In~\cite{mucha2010community}, the authors proposed a classical random walk with uniform categorical coupling for multiplex networks. They assign a homogeneous weight $\omega$ for all cross-network edges connecting the same nodes in different networks. In~\cite{de2015identifying}, the authors propose a random walk method, where for each step the walker has $(1-r)$ probability to stay in the same network and $r$ probability to follow cross-network edges. However, both of them are designed for multiplex networks and cannot be applied to general multiple networks directly. Further, both $\omega$ and $r$ are set manually and stay the same for all nodes, which means these two methods only consider the global relationship between networks, which limits their applications in general and complex networks. In ~\cite{yao2018local}, the authors propose a random walk method on multi-domain networks by including a transition matrix between different networks. However, the transition matrix between networks is set manually and is independent of the starting node(s).

\stitle{Multiple Network Analysis.}  Recently, multiple networks have drawn increasing attention in the literature due to their capability in describing graph-structure data from different domains~\cite{li2018mane, ni2018co,ou2017multi,masha2018multi,zhang2020multiplex,yu2022multiplex,zhang2022role,xia2021graph}. Under the multiple network setting, a wide range of graph mining tasks have been extended to support more realistic real-life applications, including node representation learning~\cite{ni2018co, li2018mane, masha2018multi,zhang2022role}, node clustering~\cite{cheng2013flexible, liu2015robust, ou2017multi,luo2020deep}, and link prediction~\cite{zhang2020multiplex,yu2022multiplex}. For example, Multiplex Graph Neural Network is proposed to tackle the multi-behavior recommendation problem~\cite{zhang2020multiplex}.  In~\cite{yu2022multiplex}, the authors consider the relational heterogeneity within multiplex networks and propose a multiplex heterogeneous graph convolutional network (MHGCN) to learn node representations in heterogeneous networks. Compared to previous ones, MHGCN can adaptively extract useful meta-path interactions. Instead of focusing on a specific graph mining task, in this paper, we focus on designing a novel random walk model that can be used in various tasks as a fundamental component.

\section{Random Walk on Multiple Networks}
\label{section:model}
In this section, we first introduce notations. Then the reinforced updating mechanism among random walkers in \rwm is proposed.

\begin{table}
\centering
\begin{small}
\caption{Main notations}
\label{tab:symbols}
\setlength\tabcolsep{2.8pt}
\begin{tabular}{c|l}
\hline
\multicolumn{1}{c|}{\textbf{Notation}} &
  \multicolumn{1}{c}{\textbf{Definition}}\\
\hline
$K$  & The number of networks\\
$G_{i}$&The $i$\textsuperscript{th} network.\\ 
$V_i$& The node set of $G_{i}$\\
$E_i$& The edge set of $G_{i}$\\
$\mathbf{P}_{i}$&Column-normalized transition matrix of $G_{i}$.\\
$E_{i-j}$   &The cross-edge set between network $G_i$ and $G_j$\\

$\mathbf{S}_{i\rightarrow j}$ & Column-norm. cross-trans. mat. from $V_i$ to $V_j$ \\
$u_q$&A given starting node $u_q$ from $V_q$\\
$\mathbf{e}_q$&One-hot vector with only one value-1 entry for $u_q$\\
$\mathbf{x}_i^{(t)}$ & Node visit. prob. vec. of the $i$\textsuperscript{th} walker in $G_i$ at $t$\\
$\mathbf{W}^{(t)}$&Relevance weight matrix at time $t$\\
$\mathcal{P}_i^{(t)}$ & Modified trans. mat. for the $i$\textsuperscript{th} walker in $G_i$ at $t$\\
$\alpha, \lambda,\theta$&restart factor, decay factor, covering factor.\\
\bottomrule
\end{tabular}
\end{small}
\end{table}

\subsection{Notations}\label{section:notation}
Suppose that there are $K$ undirected networks, the $i$\textsuperscript{th} ($1\leq i\leq K$) network is represented by $G_i=(V_i,E_i)$ with node set $V_i$ and edge set $E_i$. We denote its transition matrix as a column-normalized matrix $\mathbf{P}_i\in \mathbb{R}^{|V_i|\times |V_i|}$. The $(v,u)$\textsuperscript{th} entry $\mathbf{P}_i(v,u)$ represents the transition probability from node $u$ to node $v$ in $V_i$. Then the $u$\textsuperscript{th} column $\mathbf{P}_i(:,u)$ is the transition distribution from node $u$ to all nodes in $V_i$. Next, we denote ${E}_{i- j}$ the cross-connections between nodes in two networks $V_i$ and $V_j$. The corresponding cross-transition matrix is $\mathbf{S}_{i\rightarrow j}\in \mathbb{R}^{|V_j|\times |V_i|}$. Then the $u$\textsuperscript{th} column $\mathbf{S}_{i\rightarrow j}(:,u)$ is the transition distribution from node $u\in V_i$ to all nodes in $V_j$. Note that $\mathbf{S}_{i\rightarrow j}\neq \mathbf{S}_{j\rightarrow i}$. And for the multiplex networks with the same node set (e.g., Fig. \ref{fig:social_example}), $\mathbf{S}_{i\rightarrow j}$ is just an identity matrix $\mathbf{I}$ for arbitrary $i,j$.

Suppose we send a random walker on $G_i$, we let $\mathbf{x}_i^{(t)}$ be the node visiting probability vector in $G_i$ at time $t$. Then the updated vector $\mathbf{x}_i^{(t+1)}=\mathbf{P}_i\mathbf{x}_i^{(t)}$ means the probability transiting among nodes in $G_i$. And $\mathbf{S}_{i\rightarrow j}\mathbf{x}_i^{(t)}$ is the probability vector propagated into $G_j$ from $G_i$. Important notations are summarized in Table~\ref{tab:symbols}.

\subsection{Reinforced Updating Mechanism}
\label{section:updating}
In \rwm, we send out one random walker for each network. 
Initially, for the starting node $u_q\in G_q$, the corresponding walker's $\mathbf{x}_q^{(0)}=\mathbf{e}_q$ where $\mathbf{e}_q$ is a one-hot vector with only one value-1 entry corresponding to $u_q$. For other networks $G_i(i \neq q)$,  we initialize $\mathbf{x}_i^{(0)}$ by 
$\mathbf{x}_i^{(0)} = \mathbf{S}_{q\rightarrow i}\mathbf{x}_{q}^{(0)}$. That is\footnote{If there are no direct connections from $u_q$ to nodes in $V_i(i\neq q)$, we first propagate probability to other nodes in $G_q$ from $u_q$ via breadth-first-search layer by layer until we reach a node which has cross-edges to nodes in $V_i$. Then we initialize $\mathbf{x}_i^{(0)}=\mathbf{S}_{q\rightarrow i}\mathbf{P}_q^t\mathbf{e}_q$, where $t$ is the number of hops that we first reach the effective node from $u_q$ in $G_q$.}:
\begin{equation}
\begin{small}
 \label{eq:ini_x0}
        \mathbf{x}^{(0)}_i = \begin{cases}
                        \mathbf{e}_q &\text{if $i = q$} \\
                        \mathbf{S}_{q\rightarrow i}\mathbf{e}_q &\text{otherwise}
                    \end{cases}
\end{small}
\end{equation}

To update $i$\textsuperscript{th} walker's vector $\mathbf{x}_i^{(t)}$, the walker not only follows the transition probabilities in the corresponding network $G_i$, but also obtains influences from other networks. Intuitively, networks sharing relevant visited local structures should influence each other with higher weights. And the ones with less relevant visited local structures have fewer effects. We measure the relevance of visited local structures of two walkers in $G_i$ and $G_j$ with the cosine similarity of their vectors $\mathbf{x}_i^{(t)}$ and  $\mathbf{x}_j^{(t)}$. Since different networks consist of different node sets, we define the relevance as $\textrm{cos} (\mathbf{x}_{i}^{(t)},\mathbf{S}_{j\rightarrow i}\mathbf{x}_{j}^{(t)})$. Notice that when $t$ increases, walkers will explore nodes further away from the starting node. Thus, we add a decay factor $\lambda$ $(0<\lambda<1)$ in the relevance to emphasize the similarity between local structures of two different networks within a shorter visiting range. In addition, $\lambda$ can guarantee and control the convergence of the \rwm model. Formally, we let $\mathbf{W}^{(t)}\in \mathbb{R}^{K\times K}$ be the local relevance matrix among networks at time $t$ and we initialize it with identity matrix $\mathbf{I}$. We update each entry $\mathbf{W}^{(t)}(i,j)$ as follow: 
\begin{equation}
\begin{small}
 \label{eq:ws}
     \mathbf{W}^{(t)}(i,j)=\mathbf{W}^{(t-1)}(i,j)+\lambda^{t} \cos (\mathbf{x}_{i}^{(t)},\mathbf{S}_{j\rightarrow i}\mathbf{x}_{j}^{(t)})
\end{small}
\end{equation}

For the $i$\textsuperscript{th} walker, influences from other networks are reflected in the dynamic modification of the original transition matrix $\mathbf{P}_i$ based on the relevance weights. Specifically, the modified transition matrix of $G_i$ is:
\begin{equation}
\begin{small}
 \label{eq:general_p}
      \mathcal{P}_i^{(t)} = \sum_{j=1}^K \hat{\mathbf{W}}^{(t)}(i,j) \mathbf{S}_{j\rightarrow i} \mathbf{P}_j  \mathbf{S}_{i\rightarrow j}
\end{small}
 \end{equation}
where $\mathbf{S}_{j\rightarrow i} \mathbf{P}_j  \mathbf{S}_{i\rightarrow j}$ represents the propagation flow pattern $G_i\rightarrow G_j\rightarrow G_i$ (counting from the right side) and $\hat{\mathbf{W}}^{(t)}(i,j) = \frac{\mathbf{W}^{(t)}(i,j)}{\sum_k\mathbf{W}^{(t)}(i,k)}$ is the row-normalized local relevance weights from $G_j$ to $G_i$. To guarantee the stochastic property of the transition matrix, we also column-normalize $\mathcal{P}_i^{(t)}$ after each update.

At time $t+1$, the visiting probability vector of the walker on $G_i$ is updated:
\begin{equation}
\begin{small}
 \label{eq:update}
      \mathbf{x}_i^{(t+1)}=\mathcal{P}_i^{(t)} \mathbf{x}_i^{(t)}
\end{small}
\end{equation}

Different from the classic random walk model, transition matrices in \rwm dynamically evolve with local relevance influences among walkers from multiple networks. As a result, the time-dependent property enhances \rwm with the power of aggregating relevant and useful local structures among networks.%

Next, we theoretically analyze the convergence properties. First, in Theorem \ref{th:converge_general}, we present the weak convergence property~\cite{bian2017many} of the modified transition matrix $\mathcal{P}_i^{(t)}$. The convergence of the visiting probability vector will be provided in Theorem \ref{th:converge_score}.

 \begin{theorem}
 \label{th:converge_general}
 When applying \rwm on multiple networks, for any small tolerance $0<\epsilon <1$, for all $i$, when $t > \lceil  \log_\lambda\frac{\epsilon}{K^2(|V_i|+2)} \rceil$,  $\parallel {\mathcal{P}}_i^{(t+1)}-{\mathcal{P}}_i^{(t)}\parallel_\infty< \epsilon$, where $V_i$ is the node set of network $G_i$ and $\parallel \cdot \parallel_\infty$ is the $\infty$-norm of a matrix. 
 \end{theorem}

For the multiple networks with the same node set (i.e., multiplex networks), we have a faster convergence rate.

 \begin{theorem}
 \label{th:converge_p}
 When applying \rwm on multiple networks with the same node set, for any small tolerance $0<\epsilon <1$, for all $i$, $\parallel \mathcal{P}_i^{(t+1)}-\mathcal{P}_i^{(t)} \parallel_\infty< \epsilon$,  when $t >  \lceil \log_\lambda\frac{\epsilon}{K}\rceil$. \end{theorem}

\begin{proof}
In the multiplex networks, cross-transition matrices are just $\mathbf{I}$, so the stochastic property of  $\mathcal{P}_i^{(t)}$ can be naturally guaranteed without the column-normalization of $\mathcal{P}_i^{(t)}$. We then define $\Delta(t+1)=\parallel \mathcal{P}_i^{(t+1)}-\mathcal{P}_i^{(t)}\parallel_\infty$.

Based on Eq.~(\ref{eq:general_p}) and Eq.~(\ref{eq:ws}), we have 
\begin{equation*}
\begin{small}
\begin{aligned}
 \Delta(t+1)=\parallel&\mathcal{P}_i^{(t+1)}-\mathcal{P}_i^{(t)}\parallel_\infty\\
  =\parallel& \sum_{j=0}^K[\hat{\mathbf{W}}^{(t+1)}(i,j)-\hat{\mathbf{W}}^{(t)}(i,j)]\mathbf{P}_j \parallel_\infty\\
  =\parallel& \sum_{j\in L_i}[\hat{\mathbf{W}}^{(t+1)}(i,j)-\hat{\mathbf{W}}^{(t)}(i,j)]\mathbf{P}_j\\
  &+\sum_{j\in \bar{L}_i}[\hat{\mathbf{W}}^{(t+1)}(i,j)-\hat{\mathbf{W}}^{(t)}(i,j)]\mathbf{P}_j \parallel_\infty\\
\end{aligned}
\end{small}
\end{equation*}
where $L_i=\{j|\hat{\mathbf{W}}^{(t+1)}(i,j)>=\hat{\mathbf{W}}^{(t)}(i,j)\}$, and $\bar{L}_i=\{1,2...,K\}-L_i$. Since all entries in $\mathbf{P}_j$ are non-negative, for all $j$, all entries in the first part are non-negative and all entries in the second part is non-positive. Thus, we have 
\begin{equation*}
\begin{small}
\begin{aligned}
  \Delta(t+1)=\text{max}&\{\parallel \sum_{j\in L_i}[\hat{\mathbf{W}}^{(t+1)}(i,j)-\hat{\mathbf{W}}^{(t)}(i,j)]\mathbf{P}_j \parallel_\infty,\\
    &\quad \parallel \sum_{j\in \bar{L}_i}[\hat{\mathbf{W}}^{(t+1)}(i,j)-\hat{\mathbf{W}}^{(t)}(i,j)]\mathbf{P}_j \parallel_\infty\}\\
\end{aligned}
\end{small}
\end{equation*}

Since for all $i$, score vector $\mathbf{x}^{(t)}_i$ is non-negative, we have $\text{cos}(\mathbf{x}^{(t)}_i,\mathbf{x}^{(t)}_k) \geq 0$. Thus, $\sum_k \mathbf{W}^{(t+1)}(i,k) \geq \sum_k \mathbf{W}^{(t)}(i,k)$. For the first part, we have
\begin{equation*}
\begin{small}
\begin{aligned}
    &\parallel\sum_{j\in L_i}[\hat{\mathbf{W}}^{(t+1)}(i,j)-\hat{\mathbf{W}}^{(t)}(i,j)]\mathbf{P}_j \parallel_\infty \\
     =&\parallel\sum_{j\in L_i}[\frac{\mathbf{W}^{(t+1)}(i,j)}{\sum_k \mathbf{W}^{(t+1)}(i,k)}-\frac{\mathbf{W}^{(t)}(i,j)}{\sum_k \mathbf{W}^{(t)}(i,k)}]\mathbf{P}_j \parallel_\infty\\
    \leq& \parallel\sum_{j\in L_i}[\frac{\mathbf{W}^{(t)}(i,j)+\lambda^t}{\sum_k \mathbf{W}^{(t)}(i,k)}-\frac{\mathbf{W}^{(t)}(i,j)}{\sum_k \mathbf{W}^{(t)}(i,k)}]\mathbf{P}_j \parallel_\infty\\
    \leq& \parallel\sum_{j\in L_i}\frac{\lambda^t}{\sum_k \mathbf{W}^{(t)}(i,k)}\mathbf{P}_j \parallel_\infty\\
    \leq&\parallel\sum_{j\in L_i}\lambda^t \mathbf{P}_j \parallel_\infty\\
    \leq& \lambda^tK
\end{aligned}
\end{small}
\end{equation*}
Similarly, we can prove the second part has the same bound.
\begin{equation*}
\begin{small}
\begin{aligned}
&\parallel\sum_{j\in \bar{L}_i}[\hat{\mathbf{W}}^{(t+1)}(i,j)-\hat{\mathbf{W}}^{(t)}(i,j)]\mathbf{P}_j \parallel_\infty \\
=&\parallel\sum_{j\in \bar{L}_i}[\frac{\mathbf{W}^{(t)}(i,j)}{\sum_k \mathbf{W}^{(t)}(i,k)}-\frac{\mathbf{W}^{(t+1)}(i,j)}{\sum_k \mathbf{W}^{(t+1)}(i,k)}]\mathbf{P}_j \parallel_\infty\\
=&\quad \begin{aligned}[t]
 \parallel \sum_{j\in \bar{L}_i}&[\frac{\mathbf{W}^{(t)}(i,j)}{\sum_k \mathbf{W}^{(t)}(i,k)}\\
    -&\frac{\mathbf{W}^{(t)}(i,j)+\lambda^{(t+1)}\text{cos}(\mathbf{x}_{i}^{(t+1)},\mathbf{x}_{j}^{(t+1)})}{\sum_k   
    \mathbf{W}^{(t)}(i,k)+\lambda^{(t+1)}\text{cos}(\mathbf{x}_{i}^{(t+1)},\mathbf{x}_{k}^{(t+1)})}]\mathbf{P}_j \parallel_\infty\\
\end{aligned}\\
\leq& \lambda^t | \frac{\sum_k\text{cos}(\mathbf{x}_{i}^{(t+1)},\mathbf{x}_{k}^{(t+1)}) -\sum_{j \in \bar{L}_i} \text{cos}(\mathbf{x}_{i}^{(t+1)},\mathbf{x}_{j}^{(t+1)})}
     { \sum_k\mathbf{W}^{(t)}(i,k)}|\\ 
     \leq& \lambda^t K
\end{aligned}
\end{small}
\end{equation*}
Thus, we have
\begin{equation*}
\begin{small}
\begin{aligned}
  \Delta(t+1)
    \leq \lambda^tK
\end{aligned}
\end{small}
\end{equation*}
Then, we know $\Delta(t+1) \leq \epsilon $ when $t > \lceil \log_\lambda\frac{\epsilon}{K}\rceil.$

Next, we consider the general case. In Theorem ~\ref{th:converge_general}, we discuss the weak convergence of the modified transition matrix $\mathcal{P}_i^{(t)}$ in general multiple networks. Because after each iteration, $\mathcal{P}_i^{(t)}$ needs to be column-normalized to keep the stochastic property, we let $\hat{\mathcal{P}}_i^{(t)}$ represent the column-normalized one. Then Theorem ~\ref{th:converge_general} is for the residual $\Delta(t+1)=\parallel\hat{\mathcal{P}}_i^{(t+1)}-\hat{\mathcal{P}}_i^{(t)}\parallel_\infty$.

 \begin{equation*}
 \begin{small}
     \begin{aligned}
         \Delta(t+1)=&\parallel \hat{\mathcal{P}}_i^{(t+1)}-\hat{\mathcal{P}}_i^{(t)}\parallel_\infty\\
         =&\textrm{max}\{|\hat{\mathcal{P}}_i^{(t+1)}(x,y)-\hat{\mathcal{P}}_i^{(t)}(x,y)|\}\\
         =&\textrm{max}\{|\frac{\mathcal{P}_i^{(t+1)}(x,y)}{\sum_z \mathcal{P}_i^{(t+1)}(z, y)}-\frac{\mathcal{P}_i^{(t)}(x,y)}{\sum_z \mathcal{P}_i^{(t)}(z, y)}|\}
     \end{aligned}
\end{small}
 \end{equation*}

Based on Eq.~\eqref{eq:ws}, we have for all $t$, $\frac{1}{K}\leq \sum_z \mathcal{P}_i^{(t)}(z,y)\leq1$. 

We denote $\textrm{min}\{ \mathcal{P}_i^{(t+1)}(x,y),  \mathcal{P}_i^{(t)}(x,y)\}$ as $p$ and $\textrm{min}\{ \sum_z \mathcal{P}_i^{(t+1)}(z,y), \sum_z \mathcal{P}_i^{(t)}(z,y)\}$ as $m$. From the proof of Theorem~\ref{th:converge_p} , we know that 
\[| \mathcal{P}_i^{(t+1)}(x,y)-  \mathcal{P}_i^{(t)}(x,y)|\leq \lambda^t K\] 
and \[|\sum_z \mathcal{P}_i^{(t+1)}(z,y)-\sum_z \mathcal{P}_i^{(t)}(z,y)| \leq \lambda^t K|V_i|\]

Then, we have
 \begin{equation*}
 \begin{small}
     \begin{aligned}
         &|\frac{\mathcal{P}_i^{(t+1)}(x,y)}{\sum_z \mathcal{P}_i^{(t+1)}(z,y)}-\frac{\mathcal{P}_i^{(t)}(x,y)}{\sum_z \mathcal{P}_i^{(t)}(z,y)}|\\
        \leq& \frac{p+\lambda^tK}{m}-\frac{p}{m+\lambda^tK|V_i|}\\
       \leq&  \frac{m\lambda^tK+|V_i|(\lambda^tK)^2+pK|V_i|\lambda^t}{m^2}\\
       =& \lambda^t \frac{mK+K^2|V_i|\lambda^t+pK|V_i|}{m^2} 
     \end{aligned}
\end{small}
 \end{equation*}
 
When $t>\lceil -\log_\lambda (K^2|V_i|)\rceil$, we have $\lambda^tK^2|V_i| \leq 1$. Thus,
  \begin{equation*}
     \begin{small}
     \begin{aligned}
     \Delta(t+1)\leq
       &\lambda^t \frac{mK+K^2|V_i|\lambda^t+pK|V_i|}{m^2}
       \leq \lambda^t K^2(|V_i|+2)
     \end{aligned}
     \end{small}
 \end{equation*}
Then, it's derived that $\Delta(t+1) \leq \epsilon $ when $t >  \lceil  \log_\lambda\frac{\epsilon}{K^2(|V_i|+2)} \rceil$.
 \end{proof}

\subsection{\rwm With Restart Strategy}
\label{section:rwrw}
RWM is a general random walk model for multiple networks and can be further customized into different variations. In this section, we integrate the idea of random walk with restart (RWR)~\cite{tong2006fast} into RWM. 

In RWR, at each time point, the random walker explores the network based on topological transitions with $\alpha (0<\alpha<1)$ probability and jumps back to the starting node with probability $1-\alpha$. The restart strategy enables RWR to obtain proximities of all nodes to the starting node.

Similarly, we apply the restart strategy for \rwm in updating visiting probability vectors. For the $i$\textsuperscript{th} walker, we have:
\begin{equation}
    \begin{small}
    \label{eq:rwr_general}
    \mathbf{x}_{i}^{(t+1)} =\alpha {\mathcal{P}}_i^{(t)} \mathbf{x}_i^{(t)}+(1-\alpha)\mathbf{x}_i^{(0)}
    \end{small}
\end{equation}
where $\mathcal{P}_i^{(t)}$ is obtained by Eq. \eqref{eq:general_p}. Since the restart component does not provide any information for the visited local structure, we dismiss this part when calculating local relevance weights. Therefore, we modify the $\text{cos}(\cdot,\cdot)$ in Eq. \eqref{eq:ws} and have:
\begin{equation}
    \begin{small}
    \label{eq:w_rwr}
    \begin{aligned}
    \mathbf{W}^{(t)}(i,j)&=\mathbf{W}^{(t-1)}(i,j)+\\
    &\lambda^{t} \textrm{cos} ((\mathbf{x}_{i}^{(t)}-(1-\alpha)\mathbf{x}_i^{(0)}),\mathbf{S}_{j\rightarrow i}(\mathbf{x}_{j}^{(t)}-(1-\alpha)\mathbf{x}_{j}^{(0)}))
    \end{aligned}
    \end{small}
\end{equation}

 \begin{theorem}
 \label{th:converge_score}
 Adding the restart strategy in \rwm does not affect the convergence {properties} in Theorems \ref{th:converge_general} and \ref{th:converge_p}. The visiting vector $\mathbf{x}_i^{(t)}(1\leq i\leq K)$ will also converge. 
 \end{theorem}

We skip the proof here. The main idea is that $\lambda$ first guarantees the weak convergence of $\mathcal{P}_i^{(t)}$ ($\lambda$ has the same effect as in Theorem \ref{th:converge_general}). After obtaining a converged $\mathcal{P}_i^{(t)}$, Perron-Frobenius theorem~\cite{LZ:StoPro} can guarantee the convergence of $\mathbf{x}_i^{(t)}$ with a similar convergence proof of the traditional RWR model.


\subsection{Time complexity of \rwm.}
As a basic method, we can iteratively update vectors until convergence or stop the updating at a given number of iterations $T$. Algorithm \ref{alg:two_multiplex} shows the overall process. 

In each iteration, for the $i$\textsuperscript{th} walkers (line 4-5),
we update $\mathbf{x}_i^{(t+1)}$ based on Eq. \eqref{eq:general_p} and \eqref{eq:rwr_general} (line 4). Note that we do not compute the modified transition matrix $\mathcal{P}_i^{(t)}$ and output. If we substitute Eq. \eqref{eq:general_p} in Eq. \eqref{eq:rwr_general}, we have a probability propagation  $\mathbf{S}_{j\rightarrow i} \mathbf{P}_j  \mathbf{S}_{i\rightarrow j}\mathbf{x}_i^{(t)}$ which reflects the information flow $G_i\rightarrow G_j\rightarrow G_i$. In practice, a network is stored as an adjacent list. So we only need to update the visiting probabilities of direct neighbors of visited nodes and compute the propagation from right to left. Then calculating $\mathbf{S}_{j\rightarrow i} \mathbf{P}_j  \mathbf{S}_{i\rightarrow j}\mathbf{x}_i^{(t)}$ costs $O(|V_i|+|E_{i-j}|+|E_j|)$. And the restart addition in Eq. \eqref{eq:rwr_general} costs $O(|V_i|)$. As a result, line 4 costs $O(|V_i|+|E_i|+\sum_{j\neq i}(|E_{i-j}|+|E_j|))$ where $O(|E_i|)$ is from the propagation in $G_i$ itself. 

In line 5, based on Eq. \eqref{eq:w_rwr}, it takes $O(|E_{i-j}|+|V_j|)$ to get $\mathbf{S}_{j\rightarrow i}\mathbf{x}_{j}^{(t)}$ and $O(|V_i|)$ to compute cosine similarities. Then updating ${\mathbf{W}}(i,j)$ (line 5) costs $O(|E_{i-j}|+|V_i|+|V_j|)$. 
In the end, normalization (line 6) costs $O(K^2)$ which can be ignored.

In summary, with $T$ iterations, power iteration method for \rwm costs $O((\sum_i |V_i|+\sum_i|E_i|+\sum_{i\neq j} |E_{i-j}|)KT)$. For the multiplex networks with the same node set $V$, the complexity shrinks to $O((\sum_i |E_i|+K|V|)KT)$ because of the one-to-one mappings in $E_{i-j}$.

Note that power iteration methods may propagate probabilities to the entire network. Next, we present two speeding-up strategies to restrict the probability propagation into only a small subgraph around the starting node.

\begin{algorithm}[t!]
\begin{small}
\caption{Random walk on multiple networks}
\label{alg:two_multiplex}
\KwIn{Transition matrices $\{\mathbf{P}_i\}_{i=1}^K$, cross-transition matrices $\{\mathbf{S}_{i\rightarrow j}\}_{i\neq j}$,  tolerance $\epsilon$, starting node $u_q\in G_q$, decay factor $\lambda$, restart factor $\alpha$, iteration number $T$}
\KwOut{Visiting prob. vec. for $K$ walkers ($\{\mathbf{x}_i^{(T)}\}_{i=1}^K$) }
\BlankLine

Initialize $\{\mathbf{x}_i^{(0)}\}_{i=1}^K$ based on Eq. \eqref{eq:ini_x0}, and $\mathbf{W}^{(0)}=\hat{\mathbf{W}}^{(0)}=\mathbf{I}$\;

\For{$t \leftarrow 0$ \KwTo $T$}{
    \For{$i\leftarrow 1$ \KwTo $K$}{
        calculate $\mathbf{x}_{i}^{(t+1)}$ according to Eq. \eqref{eq:general_p} and \eqref{eq:rwr_general}\;
        calculate $\mathbf{W}^{(t+1)}$ according to Eq. \eqref{eq:w_rwr}\;
    }
    $\hat{\mathbf{W}}^{(t+1)}\leftarrow$ row-normalize $\mathbf{W}^{(t+1)}$\;

}
\end{small}
\end{algorithm}
\section{Speeding Up}
\label{section:computation}
In this section, we introduce two approximation methods that not only dramatically improve computational efficiency but also guarantee performance.

\subsection{Early Stopping}\label{A1}
With the decay factor $\lambda$, the modified transition matrix of each network converges before the visiting score vector does. Thus, we can approximate the transition matrix with early stopping by splitting the computation into two phases. In the first phase, we update both transition matrices and score vectors, while in the second phase, we keep the transition matrices static and only update score vectors.

Now we give the error bound between the early-stopping updated transition matrix and the power-iteration updated one. In $G_i$, we denote ${\mathcal{P}}_i^{(\infty)}$ the converged modified transition matrix. 

In $G_i (1\leq i\leq K)$, when properly selecting the split time $T_e$, the following theorem demonstrates that we can securely approximate the power-iteration updated matrix ${\mathcal{P}}_i^{(\infty)}$ with ${\mathcal{P}}_i^{(T_e)}$. 
\begin{theorem}
 \label{th:error_general}
For a given small tolerance $\epsilon$, when $t>T_e=\lceil\log_\lambda \frac{\epsilon(1-\lambda)}{K^2(|V_i|+2)} \rceil$, $\parallel{\mathcal{P}}^{(\infty)}_i-{\mathcal{P}}_i^{(t)}\parallel_\infty < \epsilon$. For the multiplex networks with the same node set, we can choose $T_e=\lceil\log_\lambda \frac{\epsilon(1-\lambda)}{K}\rceil$ to get the same estimation bound $\epsilon$.
\end{theorem}

\begin{proof}
 According to the proof of Theorem~\ref{th:converge_general}, when $T_e>\lceil -\log_\lambda (K^2|V_i|)\rceil$, we have
  \begin{equation*}
  \begin{small}
  \begin{aligned}
    &\parallel {\mathcal{P}}^{(\infty)}_i-{\mathcal{P}}_i^{(T_e)}\parallel_\infty\\
    \leq& \parallel \sum\nolimits_{t=T_e}^\infty \parallel {\mathcal{P}}^{(t+1)}_i-{\mathcal{P}}_i^{(t)}\parallel_\infty\\
    \leq& \sum\nolimits_{t={T_e}}^\infty \lambda^t K^2(|V_i|+2)\\
    =&\frac{\lambda^{T_e}K^2(|V_i|+2)}{1-\lambda}
    \end{aligned}
    \end{small}
  \end{equation*}
  
  So we can select $T_e=\lceil\log_\lambda \frac{\epsilon(1-\lambda)}{K^2(|V_i|+2)} \rceil$ such that when $t>T_e$, $\parallel{\mathcal{P}}^{(\infty)}_i-{\mathcal{P}}_i^{(t)}\parallel_\infty < \epsilon$.
\end{proof}

The time complexity of the first phase is $O((\sum_i |V_i|+\sum_i|E_i|+\sum_{i\neq j} |E_{i-j}|)KT_e)$.

\begin{algorithm}[ht]
\begin{small}
\caption{Partial Updating}
\label{alg:approximate_update}
\KwIn{$\mathbf{x}_i^{(t)}$, $\mathbf{x}_i^{(0)}$, $\{\mathbf{P}_i\}_{i=1}^K$, $\{\mathbf{S}_{i\rightarrow j}\}_{i\neq j}$, $\lambda$, $\alpha$, $\theta$}
\KwOut{approximation score vector $\tilde{\mathbf{x}}_i^{(t+1)}$}
\BlankLine
initialize an empty queue $que$\;
initialize a zero vector $\mathbf{x}^{(t)}_{i0}$ with the same size with $\mathbf{x}_i^{(t)}$\;
$que$.push($Q$) where $Q$ is the node set with positive values in $\mathbf{x}_i^{(0)}$; mark nodes in $Q$ as \textit{visited}\;
$cover \leftarrow 0$\;
\While{$que$ \textsf{\upshape is not empty and} $cover<\theta$}{
$u \leftarrow que$.pop()\;
mark $u$ as visited\;
\For{\textsf{\upshape{each neighbor node $v$ of}} $u$}{
\If{$v$ \textsf{\upshape{is not marked as visited}}}{
$que$.push($v$)\;
}
}
$\mathbf{x}^{(t)}_{i0}[u] \leftarrow \mathbf{x}_i^{(t)}[u]$\;
$cover \leftarrow cover+ \mathbf{x}_i^{(t)}[u]$\;
}
Update $\tilde{\mathbf{x}}_i^{(t+1)}$ based on Eq. \eqref{eq:xapproximation}\;
\end{small}
\end{algorithm}

\subsection{Partial Updating}\label{A2}
In this section, we propose a heuristic strategy to further speed up the vector updating in Algorithm \ref{alg:two_multiplex} (line 4) by only updating a subset of nodes that covers most probabilities. 

Specifically, given a covering factor $\theta \in (0,1]$, for walker $i$, in the $t$\textsuperscript{th} iteration, we separate $\mathbf{x}_i^{(t)}$ into two non-negative vectors, $ \mathbf{x}_{i0}^{(t)}$ and $\Delta\mathbf{x}_i^{(t)}$, so that $\mathbf{x}_i^{(t)}= \mathbf{x}_{i0}^{(t)}+\Delta\mathbf{x}_i^{(t)}$, and $\parallel \mathbf{x}_{i0}^{(t)}\parallel_1 \geq \theta$. Then, we approximate $\mathbf{x}_i^{(t+1)}$ with 
\begin{align}
\begin{small}
\label{eq:xapproximation}
    \tilde{\mathbf{x}}_i^{(t+1)} = \alpha \mathcal{P}_i^{(t)}\mathbf{x}_{i0}^{(t)}+(1-\alpha \parallel\mathbf{x}_{i0}^{(t)}\parallel_1)\mathbf{x}_i^{(0)}
\end{small}
\end{align}

Thus, we replace the updating operation of $\mathbf{x}_i^{(t+1)}$  in Algorithm \ref{alg:two_multiplex} (line 4) with $\tilde{\mathbf{x}}_{i}^{(t+1)}$. The details are shown in Algorithm~\ref{alg:approximate_update}. Intuitively,  nodes close to the starting node have higher scores than nodes far away. Thus, we utilize the breadth-first search (BFS) to expand $\mathbf{x}_{i0}^{(t)}$ from the starting node set until $\parallel \mathbf{x}_{i0}^{(t)}\parallel_1 \geq \theta$ (lines 3-12). Then, in line 13,  we approximate the score vector in the next iteration according to Eq. (\ref{eq:xapproximation}).

Let $V_{i0}^{(t)}$ be the set of nodes with positive values in $\mathbf{x}_{i0}^{(t)}$, and $|E_{i0}^{(t)}|$ be the summation of out-degrees of nodes in $V_{i0}^{(t)}$. Then $|E_{i0}^{(t)}| \ll |E_i|$ dramatically reduces the number of nodes to update in each iteration. 

\section{An Exemplar application of \rwm}
\label{sec:app}
Random walk based methods are routinely used in various tasks, such as local community detection~\cite{andersen2006local,alamgir2010multi,li2015uncovering,kloster2014heat,MWC_2,bian2018multi,bian2019memory,yan2019constrained,veldt2019flow}, link prediction~\cite{martinez2017survey}, and network embedding~\cite{cao2015grarep, grover2016node2vec, perozzi2014deepwalk, tang2015line}. In this section, we take the local community detection as an example to show the application of the proposed \rwm. 
As a fundamental task in large network analysis, local community detection has attracted extensive attention recently. Unlike the time-consuming global community detection, the goal of local community detection is to detect a set of nodes with dense connections (i.e., the target local community) that contains a given query node or a query node set. Specifically, given a query node\footnote{For simplicity, the illustration is for one query node. We can easily modify our model for a set of query nodes by initializing the visiting vector $\mathbf{x}_q^{(0)}$ with uniform entry value 1/$n$ if there are $n$ query nodes.} $u_q$ in the query network $G_q$, the target of local community detection in multiple networks 
is to detect relevant local communities in all networks $G_i(1\leq i\leq K)$. 

To find the local community in network $G_{i}$, we follow the common process in the literature~\cite{andersen2006local,bian2017many,kloster2014heat,yao2018local}. By setting the query node as the starting node,  we first apply \rwm with restart strategy to calculate the converged score vector $\mathbf{x}^{(T)}_{i}$, according to which we sort nodes. $T$ is the number of iterations. Suppose there are $L$ non-zero elements in $\mathbf{x}^{(T)}_{i}$, for each $l$ $(1 \leq l \leq L)$, we compute the conductance of the subgraph induced by the top $l$ ranked nodes. The node set with the smallest conductance will be considered as the detected local community.
\section{Experiments}
\label{section:experiment}
We perform comprehensive experimental studies to evaluate the effectiveness and efficiency of the proposed methods. Our algorithms are implemented with C++.  The code and data used in this work are available.\footnote{https://github.com/flyingdoog/RWM/}
 All experiments are conducted on a PC with Intel Core I7-6700 CPU and 32 GB memory, running 64-bit Windows 10.

\subsection{Datasets}

\stitle{Datasets.} Eight real-world datasets are used to evaluate the effectiveness of the selected methods. Statistics are summarized in Table~\ref{tab:dataset}.

\stab{$\bullet$} \mitdata~\cite{kim2015community} has 10 networks, each with 91 nodes. Nodes represent phones and one edge exists if two phones detect each other under a mobile network. Each network describes connections between phones in a month. 
Phones are divided into two communities according to their owners' affiliations. 

\stab{$\bullet$} \brainnet~\cite{van2013wu} has 468 brain networks, one for each participant. In the brain network, nodes represent human brain regions and an edge depicts the functional association between two regions. Different participants may have different functional connections. Each network has 264 nodes, among which 177 nodes are studied to belong to 12 high-level functional systems, including auditory, memory retrieval, visual \textit{, etc}. Each functional system is considered a community.

 The other four are general multiple networks with different types of nodes and edges and many-to-many cross-edges between nodes in different networks.
 
\stab{$\bullet$} \sixng \& \nineng~\cite{ni2018co} are two multi-domain network datasets constructed from the 20-Newsgroup dataset. \sixng contains 5 networks of sizes \{600, 750, 900, 1050, 1200\} and \nineng consists of 5 networks of sizes \{900, 1125, 1350, 1575, 1800\}. Nodes represent news documents and edges describe their semantic similarities. The cross-network relationships are measured by cosine similarity between two documents from two networks. Nodes in the five networks in \sixng and \nineng are selected from 6 and 9 newsgroups, respectively. Each newsgroup is considered a community.

\stab{$\bullet$} \ci~\cite{lim2016bibliographic} is from an academic search engine, Citeseer. It contains a researcher collaboration network, a paper citation network, and a paper similarity network. The collaboration network has 3,284 nodes (researchers) and 13,781 edges (collaborations). The paper citation network has 2,035 nodes (papers) and 3,356 edges (paper citations). The paper similarity network has 10,214 nodes (papers) and 39,411 edges (content similarity). There are three types of cross-edges: 2,634 collaboration-citation relations, 7,173 collaboration-similarity connections, and 2,021 citation-similarity edges. 10 communities of authors and papers are labeled based on research areas.

\stab{$\bullet$} \dblp~\cite{tang2008arnetminer} consists of an author collaboration network and a paper citation network. The collaboration network has 1,209,164 nodes and 4,532,273 edges. The citation network consists of 2,150,157 papers connected by 4,191,677 citations. These two networks are connected by 5,851,893 author-paper edges. From one venue, we form an author community by extracting the ones who published more than 3 papers in that venue. We select communities with sizes ranging from 5 to 200, leading to 2,373 communities.

The above six datasets are with community information. We also include two datasets without label information for link prediction tasks only.

\stab{$\bullet$}  \athle~\cite{omodei2015characterizing}: is a directed and weighted multiplex network dataset, obtained from Twitter during an exceptional event: the 2013 World Championships in Athletics. The multiplex network makes use of 3 networks, corresponding to retweets, mentions and replies observed between 2013-08-05 11:25:46 and  2013-08-19 14:35:21. There are 88,804 nodes and (104,959, 92,370, 12,921) edges in each network, respectively.

\stab{$\bullet$} \gene~\cite{de2015structural} is a multiplex network describing different types of genetic interactions for organisms in the Biological General Repository for Interaction Datasets (BioGRID, thebiogrid.org), a public database that archives and disseminates genetic and protein interaction data from humans and model organisms. There are 7 networks in the multiplex network, including physical association, suppressive genetic interaction defined by inequality, direct interaction, synthetic genetic interaction defined by inequality, association,  colocalization, and additive genetic interaction defined by inequality. There are 6,570 nodes and 282,754 edges.

\begin{table}[ht]
    \begin{small}
    \centering
    \caption{Statistics of datasets}
    \label{tab:dataset}
    \begin{tabular}{c|c|c|c|c}
    \hline
   Dataset &\#Net. &\#Nodes  & \begin{tabular}{@{}c@{}} \#Inside \\ net. edges\end{tabular}
    &\begin{tabular}{@{}c@{}} \#Cross \\net. edges\end{tabular}\\
    \hline
    	 \mitdata   &10   &910       &14,289 & -- \\
	 \brainnet  &468    &123,552   &619,080 &--\\
	 \athle     &3      &88,804 & 210250 &-\\
	 \gene      &7      &6,570 & 28,2754 & - \\
	 \hline
	 \sixng     &5    &4,500     &9,000 &20,984 \\
	 \nineng    &5    &6,750     &13,500 &31,480\\
	 \ci        &3    &15,533    &56,548 &11,828\\
	 \dblp      &2    &3,359,321 &8,723,940 &5,851,893\\ 
    \bottomrule
    \end{tabular}
    \end{small}
\end{table}

\begin{table*}[ht!]
\centering
\begin{small}
\caption{Accuracy performances comparison}
\label{tab:ac}
\begin{tabular}{c|c|c|c|c|c|c}
\hline
Method   &\mitdata  &\brainnet &\sixng      &\nineng    &\ci    &\dblp \\
\hline
RWR      &0.621     &0.314     & 0.282	&0.261   &0.426   &0.120 \\
MWC   	 &0.588     &0.366     &0.309   &0.246   &0.367   &0.116 \\
QDC      &0.150     &0.262     &0.174   &0.147  &0.132   &0.103 \\
LEMON    &0.637     &0.266     &0.336   &0.250  &0.107   &0.090 \\
$k$-core &0.599     &0.189     &0.283	&0.200  &0.333  &0.052\\
\hline
MRWR 	 &0.671     &0.308     &0.470   &0.238  &0.398   &0.126 \\
ML-LCD   &0.361 &0.171 &- &-  &-  &- \\
\hline
\rwm(ours)   &\textbf{0.732}* & \textbf{0.403}* &\textbf{0.552}*	& \textbf{0.277}*	  &\textbf{0.478}* &\textbf{0.133}*	\\
\hline\hline
Improve &9.09\% &10.11\% & 17.4\%& 6.13\%& 12.21\%& 5.56\%\\
\hline
  \multicolumn{7}{l}{*We do student-t test between \rwm and the second best method MRWR}\\
  \multicolumn{7}{l}{(* corresponds to $p<0.05$).}\\
\end{tabular}
\end{small}
\end{table*}

\subsection{Local community detection}

\subsubsection{Experimental setup}
We compare \rwm  with seven state-of-the-art local community detection methods. RWR~\cite{tong2006fast} uses a lazy variation of random walk to rank nodes and sweeps the ranked nodes to detect local communities. MWC~\cite{bian2017many} uses the multi-walk chain model to measure node proximity scores. QDC~\cite{wu2015robust} finds local communities by extracting query biased densest connected subgraphs. LEMON~\cite{li2015uncovering} is a local spectral approach. $k$-core~\cite{kCoreEfficient} conducts graph core-decomposition and queries the community containing the query node. Note that these five methods are for single networks. The following two are for multiple networks. ML-LCD~\cite{interdonato2017local} uses a greedy strategy to find the local communities on multiplex networks with the same node set. MRWR~\cite{yao2018local} only focuses on the query network. Besides, without prior knowledge, MRWR treats other networks contributing equally.

For our method, \rwm, we adopt two approximations and set $\epsilon=0.01, \theta = 0.9$. Restart factor $\alpha$ and decay factor $\lambda$ are searched from 0.1 to 0.9. Extensive parameter studies are conducted in Sec. \ref{sec:para}. For baseline methods, we tune parameters according to their original papers and report the best results. Specifically, for RWR, MWC, and MRWR, we tune the Restart factor from 0.1 to 0.9. Other parameters in MWC are kept the default values~\cite{bian2017many}. For QDC, we follow the instruction from the original paper and search the decay factor from 0.6 to 0.95~\cite{wu2015robust}. For LEMON~\cite{li2015uncovering}, the minimum community size and expand step are set to 5 and 6, respectively.  For ML-LCD, we adopt the vanilla Jaccard similarity as the node similarity measure~\cite{interdonato2017local}.

\subsubsection{Effectiveness Evaluation}

\stitle{Evaluation on detected communities.}
For each dataset, in each experiment trial, we randomly pick one node with label information from a network as the query. Our method \rwm can detect all query-relevant communities from all networks, while all baseline methods can only detect one local community in the network containing the query node. Thus, in this section, to be fair, we only compare detected communities in the query network. In Sec. \ref{sec:case}, we also verify that \rwm can detect relevant and meaningful local communities from other networks. 

Each experiment is repeated 1000 trials, and the Macro-F1 scores are reported in Table~\ref{tab:ac}. Note that ML-LCD can only be applied to the multiplex network with the same node set. 

From Table~\ref{tab:ac}, we see that, in general, random walk based methods including RWR, MWC, MRWR, and \rwm perform better than others. It demonstrates the advance of applying random walk for local community detection.
Generally, the size of detected communities by QDC is very small, while that detected by ML-LCD is much larger than the ground truth. $k$-core suffers from finding a proper node core structure with a reasonable size, it either considers a very small number of nodes or the whole network as the detected result. Second, performances of methods for a single network, including RWR, MWC, QDC, and LEMON, are relatively low, since the single networks are noisy and incomplete. MRWR achieves the second best results on \sixng, \mitdata, and \dblp but performs worse than RWR and MWC on other datasets. Because not all networks provide equal and useful assistance for detection, treating all networks equally in MRMR may introduce noises and decrease performance. Third, our method \rwm achieves the highest F1-scores on all datasets and outperforms the second best methods by 6.13\% to 17.4\%. We conduct student-t tests between our method RWM and the second best method, MRWR. Low $p$-values $(<0.05)$ indicate that our results are statistically significant. This is because \rwm can actively aggregate information from highly relevant local structures in other networks during updating visiting probabilities. We emphasize that only \rwm can detect relevant local communities from other networks except for the network containing the query node. Please refer to Sec. \ref{sec:case} for details.

\stitle{Ranking evaluation.}
To gain further insight into why \rwm outperforms others, we compare \rwm with other random walk based methods, i.e., RWR, MWC, and MRWR, as follows. Intuitively, nodes in the target local community should be assigned with high proximity scores w.r.t. the query node. Then we check the precision of top-ranked nodes based on score vectors of those models, i.e., \textit{prec.} = ($|\text{top-}k \text{nodes} \cap \text{ground truth}|)/k$

\begin{figure*}[t]
    \centering
    \subfigure[\mitdata]{\includegraphics[width=1.7in]{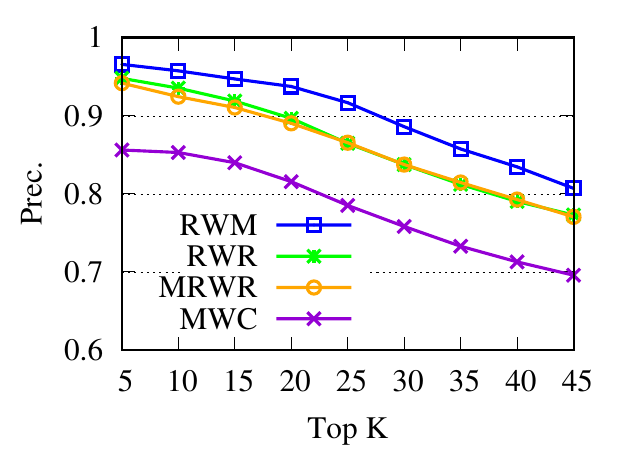}\label{fig:lcdrm}}
    \subfigure[\brainnet]{\includegraphics[width=1.71in]{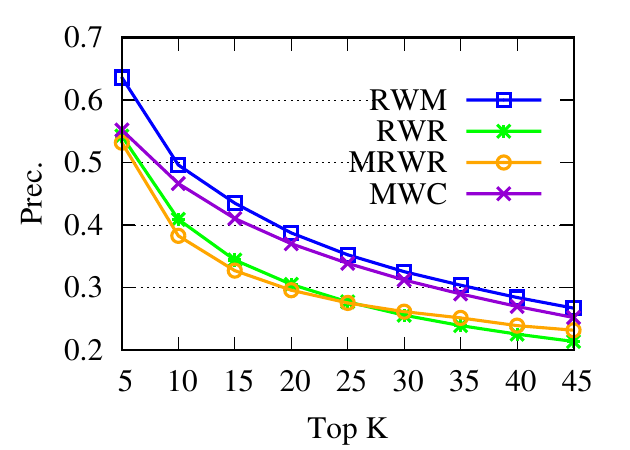}\label{fig:lcdBrain}}
    \subfigure[\sixng]{\includegraphics[width=1.71in]{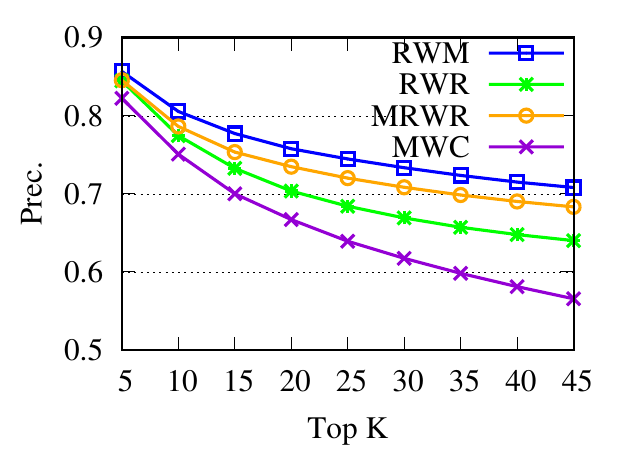}\label{fig:lcd6ng}}\\
    \subfigure[\nineng]{\includegraphics[width=1.71in]{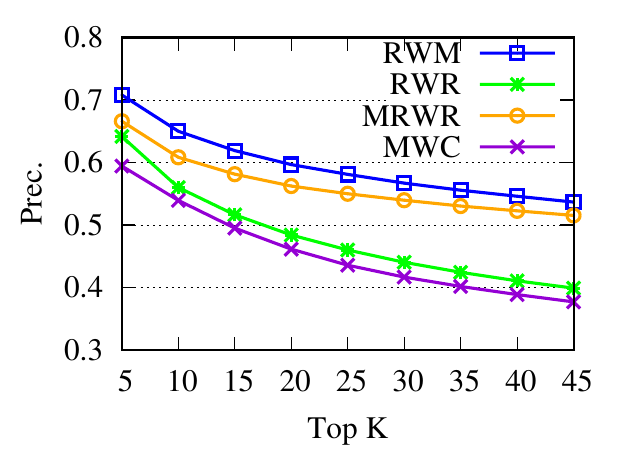}\label{fig:lcd9ng}}
    \subfigure[\ci]{\includegraphics[width=1.71in]{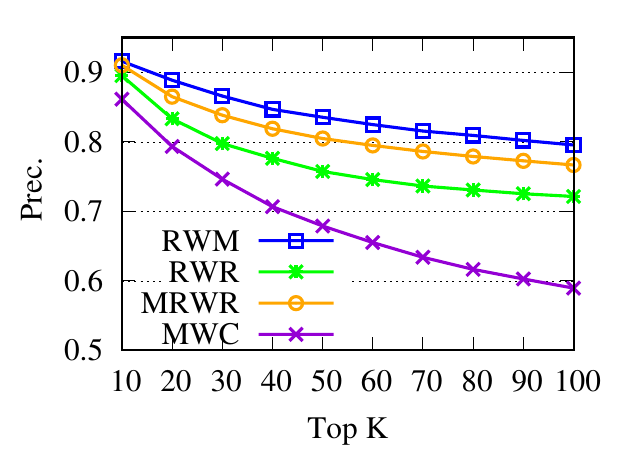}\label{fig:lcdCite}}
    \subfigure[\dblp]{\includegraphics[width=1.71in]{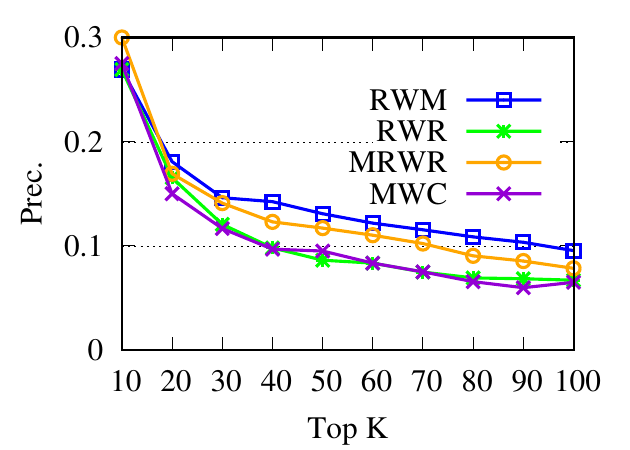}\label{fig:lcddblp}}
    \caption{Ranking evaluation of random walk based methods}
    \label{fig:lcd}
\end{figure*}
The precision results are shown in Fig.~\ref{fig:lcd}. First, we can see that the precision scores of the selected methods decrease when the size of the detected community $k$ increases. Second, our method consistently outperforms other random walk-based methods. Results indicate that \rwm ranks nodes in the ground truth more accurately. Since ranking is the basis of a random walk based method for local community detection, better node-ranking generated by \rwm leads to high-quality communities (Table \ref{tab:ac}).

\subsubsection{Parameter Study}\label{sec:para}
In this section, we show the effects of the four important parameters in \rwm: the restart factor $\alpha$, the decay factor $\lambda$, tolerance $\epsilon$ for early stopping, and the covering factor $\theta$ in the partial updating. We report the F1 scores and running time on two representative datasets \mitdata and \sixng. Note that $\epsilon$ and $\theta$ control the trade-off between running time and accuracy.

The parameter $\alpha$ controls the restart of a random walker in Eq. \eqref{eq:rwr_general}. The F1 scores and the running time \wrt $\alpha$ are shown in Fig.~\ref{fig:rmalpha} and Fig.~\ref{fig:6ngalpha}. When $\alpha$ is small, the accuracy increases as $\alpha$ increases because larger $\alpha$ encourages further exploration. When $\alpha$ reaches an optimal value, the accuracy begins to drop slightly. Because a large $\alpha$ impairs the locality property of the restart strategy. The running time increases when $\alpha$ increases because larger $\alpha$ requires more iterations for score vectors to converge.

$\lambda$ controls the updating of the relevance weights in Eq. \eqref{eq:w_rwr}.  Results in Fig.~\ref{fig:rmlambda} and Fig.~\ref{fig:6nglambda} reflect that for \mitdata, \rwm achieves the best result when $\lambda=0.9$. This is because large $\lambda$ ensures that enough neighbors are included when calculating relevance similarities. For \sixng, \rwm achieves high accuracy in a wide range of $\lambda$. For the running time, according to Theorem~\ref{th:error_general}, larger $\lambda$ results in larger $T_e$, i.e., more iterations in the first phase, and longer running time. 

$\epsilon$ controls $T_e$, the splitting time in the first phase. Instead of adjusting $\epsilon$, we directly tune $T_e$. Theoretically, a larger $T_e$ (i.e., a smaller $\epsilon$) leads to more accurate results. Based on results shown in Fig.~\ref{fig:rmepsilon} and Fig.~\ref{fig:6ngepsilon}, we notice that \rwm achieves good performance even with a small $T_e$ in the first phase. The running time decreases significantly as well. This demonstrates the rationality of early stopping.

$\theta$ controls the number of updated nodes (in Sec. \ref{A2}). In Fig.~\ref{fig:rmtheta} and \ref{fig:6ngtheta}, we see that the running time decreases along with $\theta$ decreasing because a smaller number of nodes are updated. The consistent accuracy performance shows the effectiveness of the speeding-up strategy.

\begin{figure}[t!]
\centering
\subfigure[\mitdata]{\includegraphics[width=1.7in]{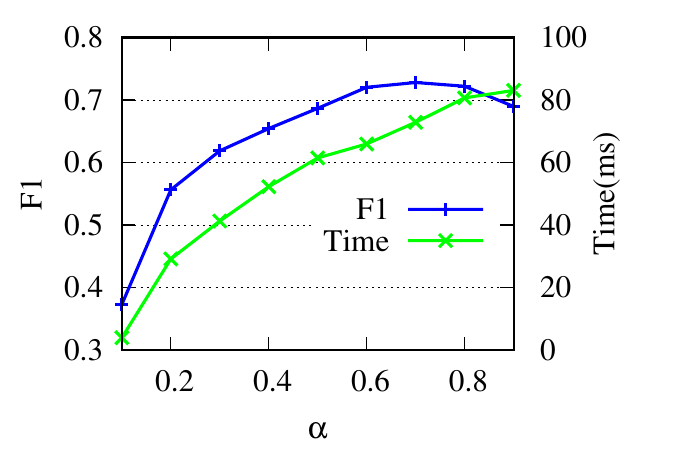}\label{fig:rmalpha}}\hspace{-1.2em}
\subfigure[\sixng]{\includegraphics[width=1.7in]{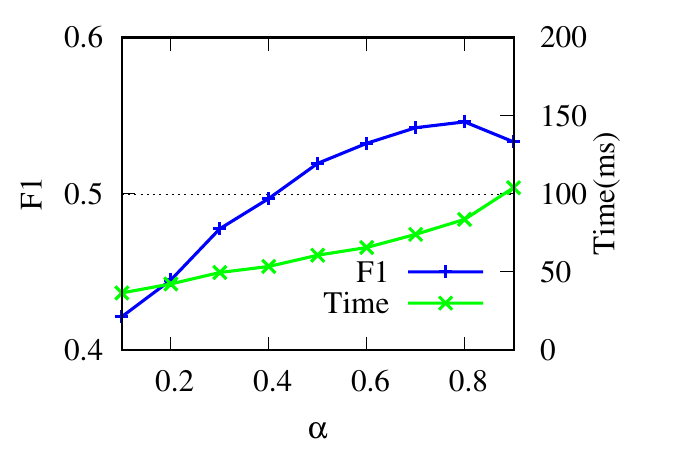}\label{fig:6ngalpha}}\hspace{-1.2em}\\
\subfigure[\mitdata]{\includegraphics[width=1.7in]{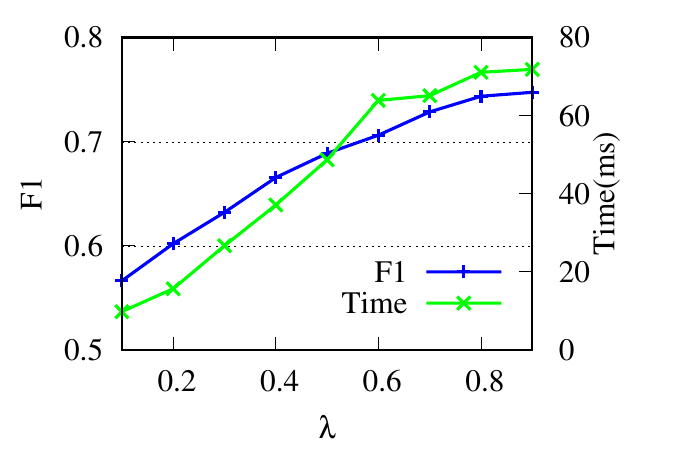}\label{fig:rmlambda}}\hspace{-1.2em}
\subfigure[\sixng]{\includegraphics[width=1.7in]{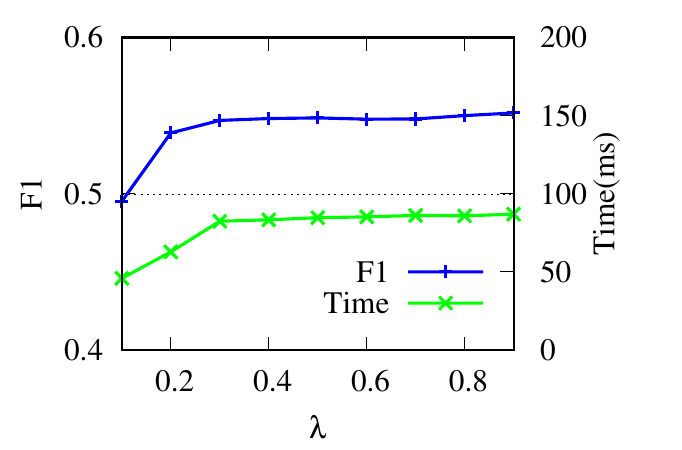}\label{fig:6nglambda}}\hspace{-1.2em}\\
\subfigure[\mitdata]{\includegraphics[width=1.7in]{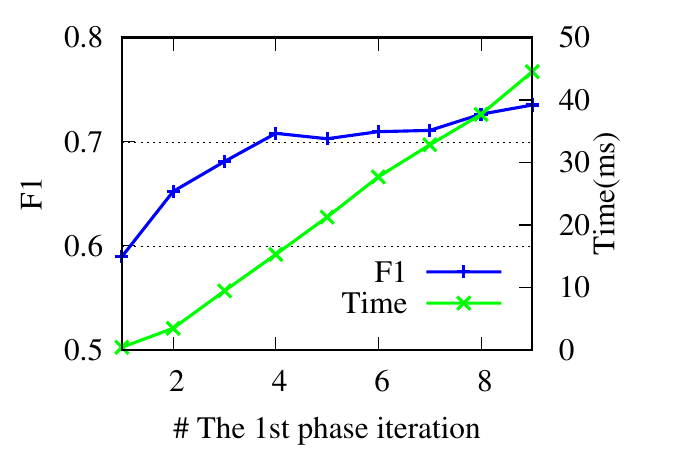}\label{fig:rmepsilon}}\hspace{-1.2em}
\subfigure[\sixng]{\includegraphics[width=1.7in]{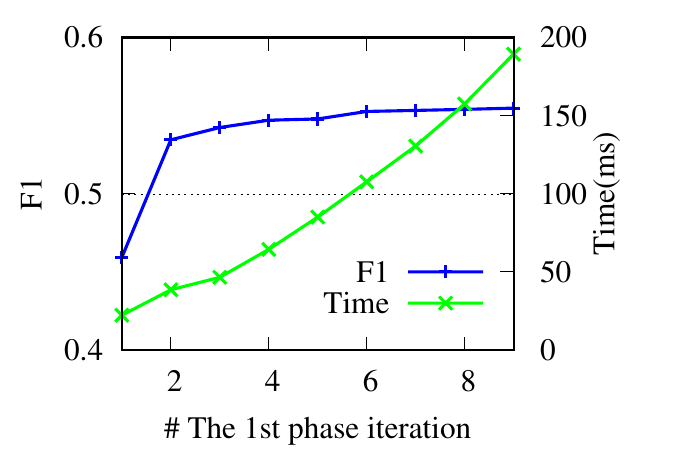}\label{fig:6ngepsilon}}\hspace{-1.2em}\\
\subfigure[\mitdata]{\includegraphics[width=1.7in]{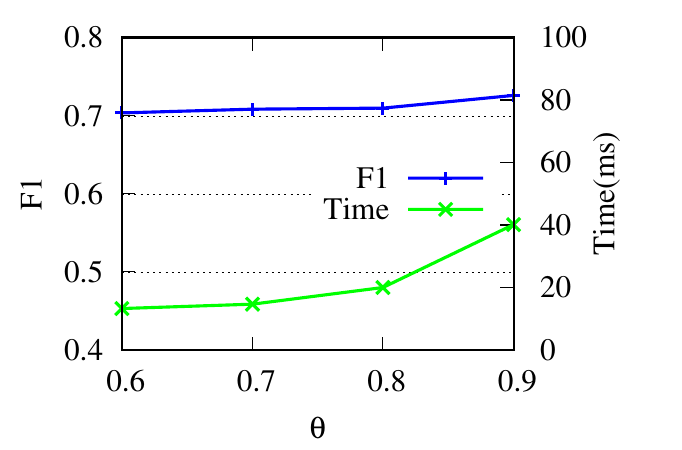}\label{fig:rmtheta}}\hspace{-1.2em}
\subfigure[\sixng]{\includegraphics[width=1.7in]{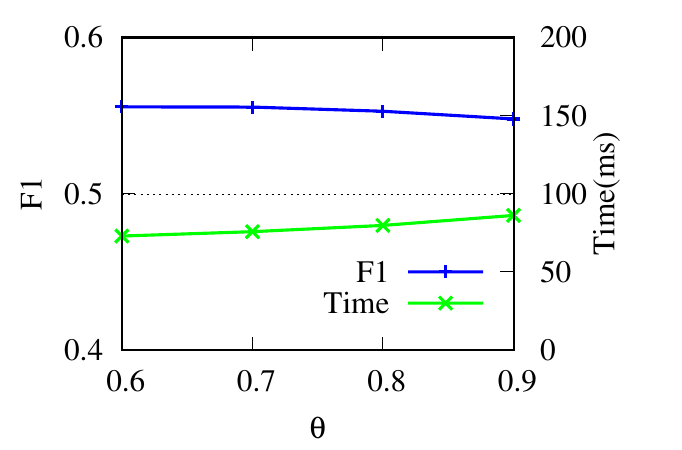}\label{fig:6ngtheta}}
\caption{Parameter study on \mitdata and \sixng}
\label{fig:para}
\end{figure}

\subsubsection{Case Studies}\label{sec:case}
\brainnet and \dblp are two representative datasets for multiplex networks (with the same node set) and the general multi-domain network (with flexible nodes and edges). We do two case studies to show the detected local communities by \rwm.
 
\stitle{Case Study on the \brainnet Dataset.}
Detecting and monitoring functional systems in the human brain is an important task in neuroscience. Brain networks can be built from neuroimages where nodes and edges are brain regions and their functional relations. In many cases, however, the brain network generated from a single subject can be noisy and incomplete. Using brain networks from many subjects helps to identify functional systems more accurately. For example, brain networks from three subjects are shown in Fig. \ref{fig:example}. Subjects 1 and 2 have similar visual conditions (red nodes); subjects 1 and 3 are with similar auditory conditions (blue nodes). For a given query region, we want to find related regions with the same functionality. 

\begin{figure}[t!]
    \centering
   \subfigure[Subject 1]{\includegraphics[width=0.95in]{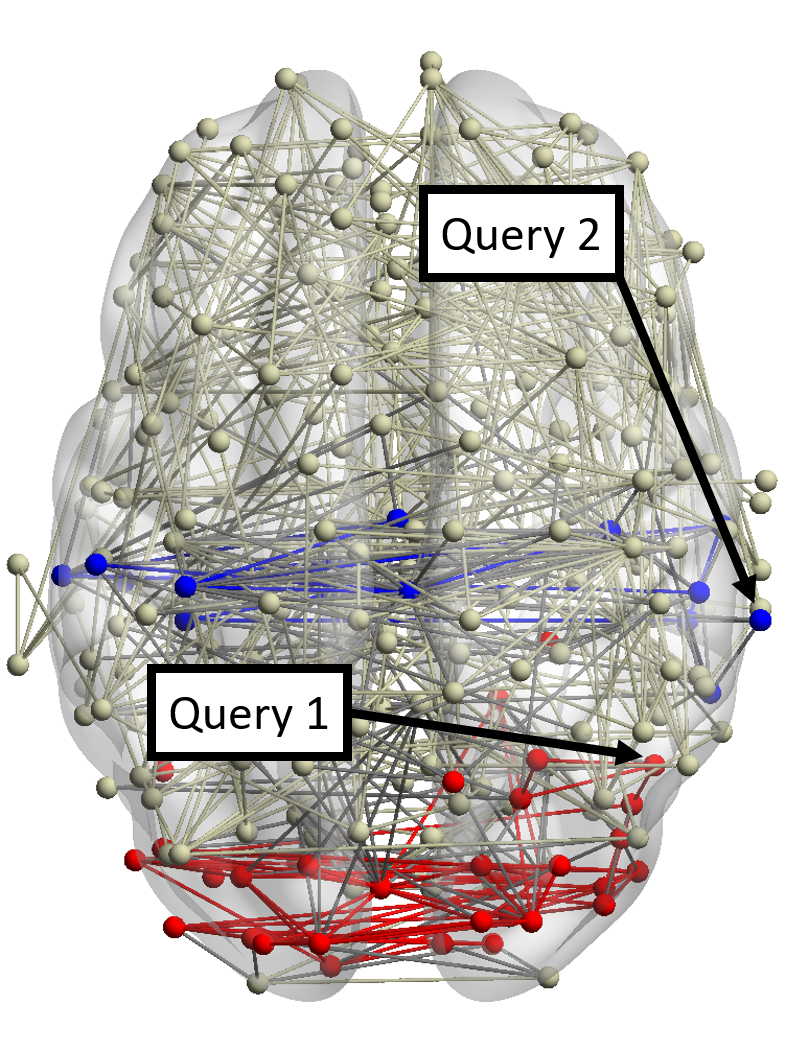}\label{fig:s0}}
   \quad
   \subfigure[Subject 2]{\includegraphics[width=0.95in]{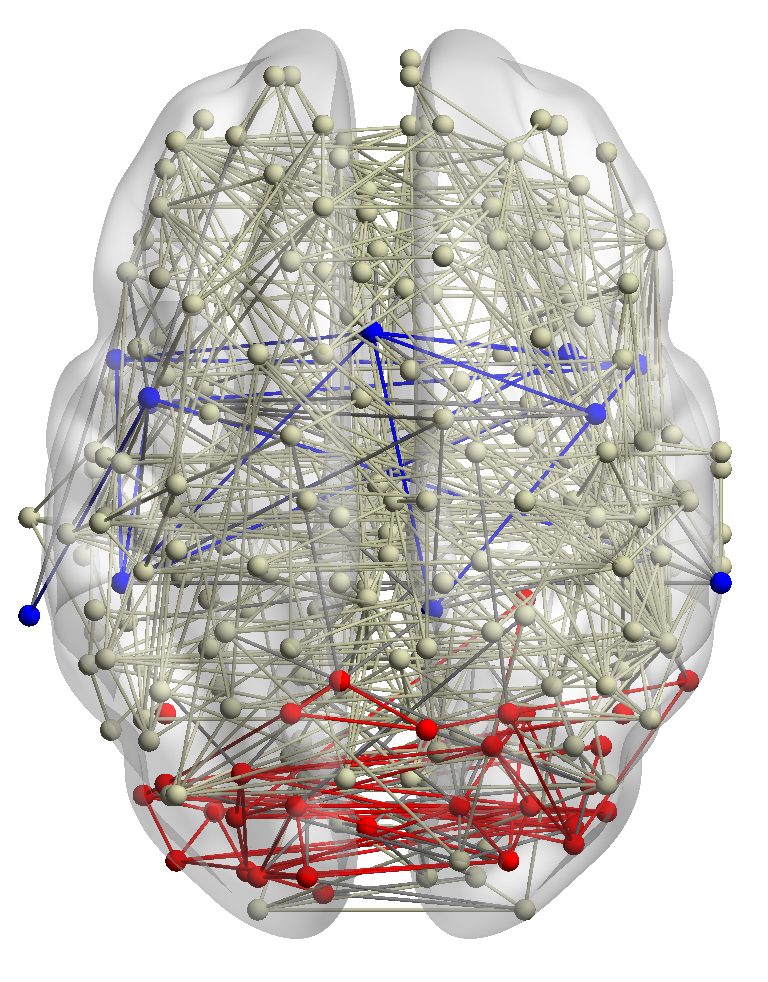}\label{fig:s1}}
   \quad
   \subfigure[Subject 3]{\includegraphics[width=0.95in]{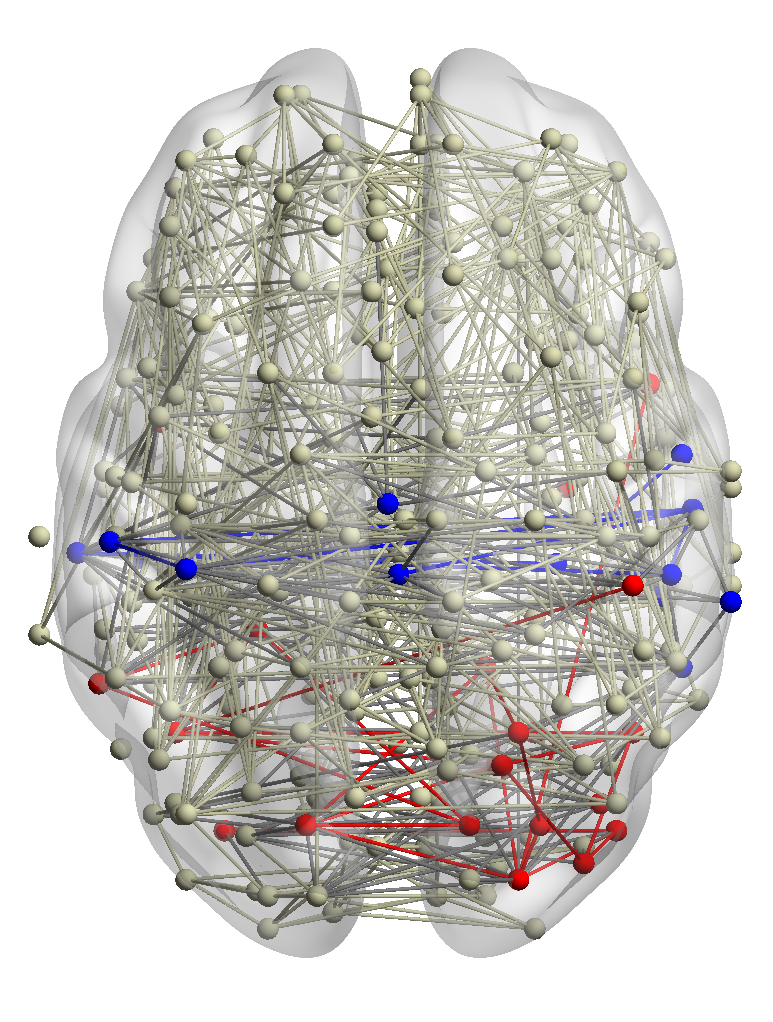}\label{fig:s2}}
 \caption{Brain networks of three subjects.}
     \label{fig:example}
\end{figure}

\stab{$\bullet$}\textbf{Detect relevant networks.} 
To see whether \rwm can automatically detect relevant networks, we run \rwm model for Query 1 and Query 2 in Fig.~\ref{fig:example} separately. Fig.~\ref{fig:weights} shows the cosine similarity between the visiting probability vectors of different walkers along iterations. $\mathbf{x}_1$, $\mathbf{x}_2$, and $\mathbf{x}_3$ are the three visiting vectors on the three brain networks, respectively. We see that the similarity between the visiting histories of walkers in relevant networks, \ie subjects 1 and 2 in Fig.~\ref{fig:w1}, subjects 1 and 3 in~\ref{fig:w2}, increases along with the updating. But similarities in query-oriented irrelevant networks are very low during the whole process. This indicates that \rwm can actively select query-oriented relevant networks to help better capture the local structure for each network.

\begin{figure}[t!]
    \centering
    \subfigure[Query 1]{\includegraphics[width=1.6in]{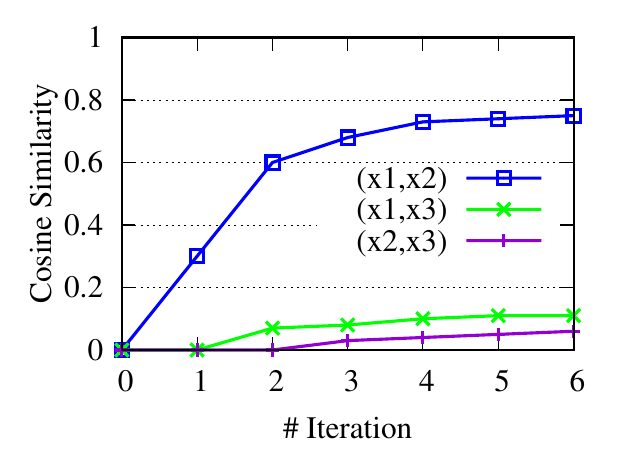}\label{fig:w1}}\hspace{-1em}
    \subfigure[Query 2]{\includegraphics[width=1.6in]{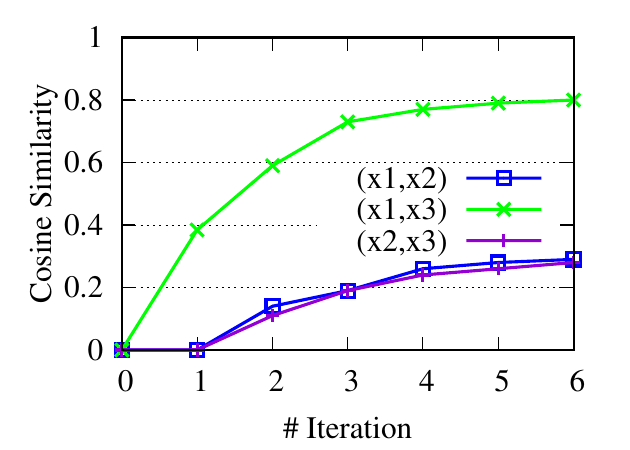}\label{fig:w2}}
    \caption{Similarity scores of walkers \wrt the iteration}
    \label{fig:weights}
\end{figure}

\stab{$\bullet$}\textbf{Identify functional systems.}
In this case study, we further evaluate the detected results of \rwm in the \brainnet dataset. Note that other methods can only find a query-oriented local community in the network containing the query node. 

Figure~\ref{fig:brain_Gt} shows the brain network of a subject in the \brainnet dataset. The nodes highlighted in red represent the suggested visual system of the human brain, which is used as the ground truth community. We can see that nodes in the visual system are not only functionally related but also spatially close. We choose a node from the visual system as the query node, which is marked in Figure~\ref{fig:brain_Gt}. We apply our method as well as other baseline methods to detect the local community in the brain network of this subject. The identified communities are highlighted in red in Figure~\ref{fig:brain_caseStudy}. From Figure~\ref{fig:brain_RWRM}, we can see that the community detected by our method is very similar to the ground truth. Single network methods, such as RWR (Figure~\ref{fig:brain_RWR}), MWC (Figure~\ref{fig:brain_MWC}), LEMON (Figure~\ref{fig:brain_LEMON}) and QDC (Figure~\ref{fig:brain_QDC}), suffer from the incomplete information in the single network. Compared to the ground truth, they tend to find smaller communities. MRWR (Figure~\ref{fig:brain_RWRM}) includes many false positive nodes. The reason is that MRWR assumes all networks are similar and treat them equally. ML-LCD (Figure~\ref{fig:mllcd}) achieves a relatively reasonable detection, while it still neglects nodes in the boundary area.

 \begin{figure}[ht]
 \centering
\subfigure[Groud Truth]{\includegraphics[width=1.2in]{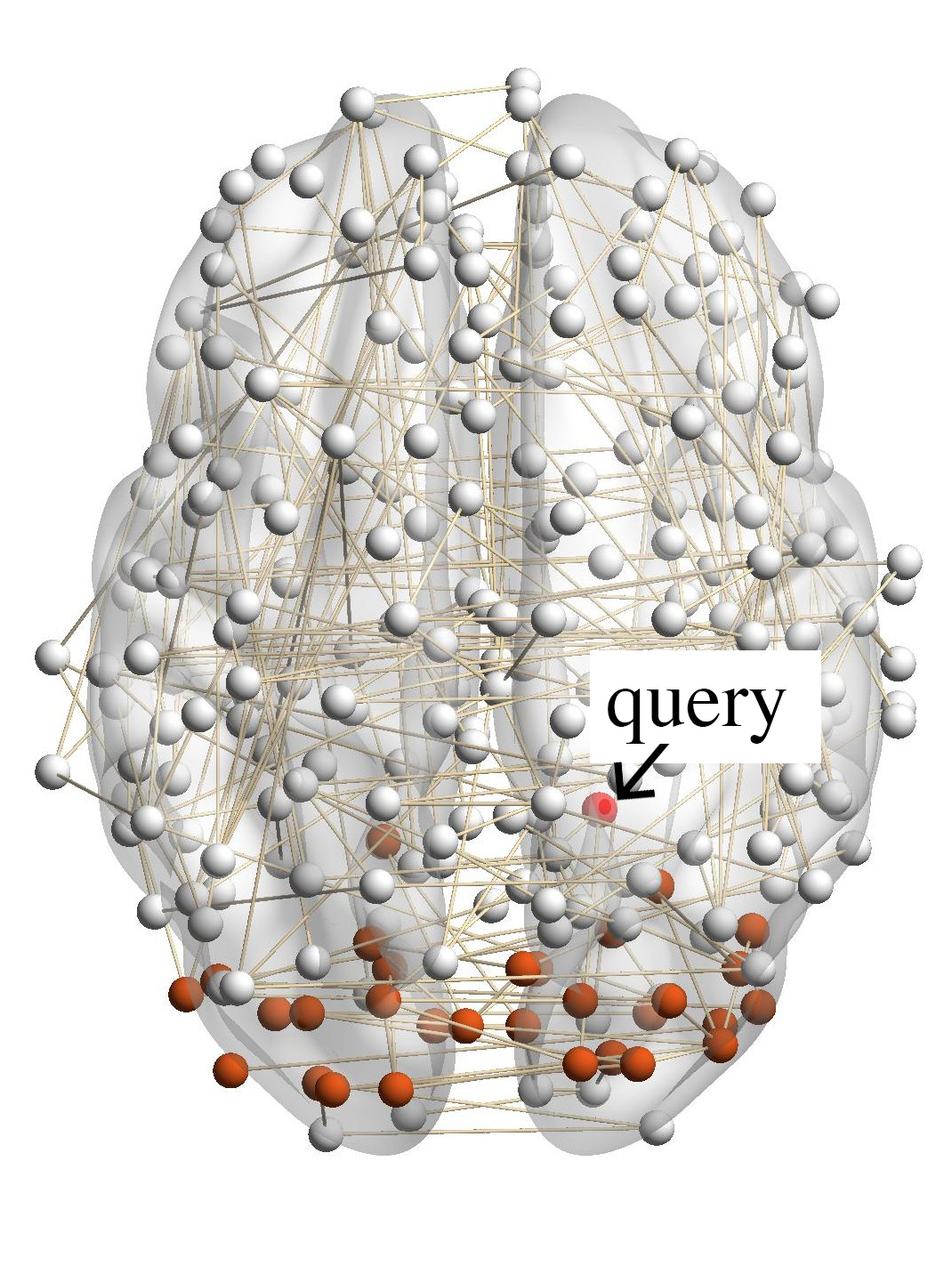}\label{fig:brain_Gt}}
\subfigure[\rwm]{\includegraphics[width=1.2in]{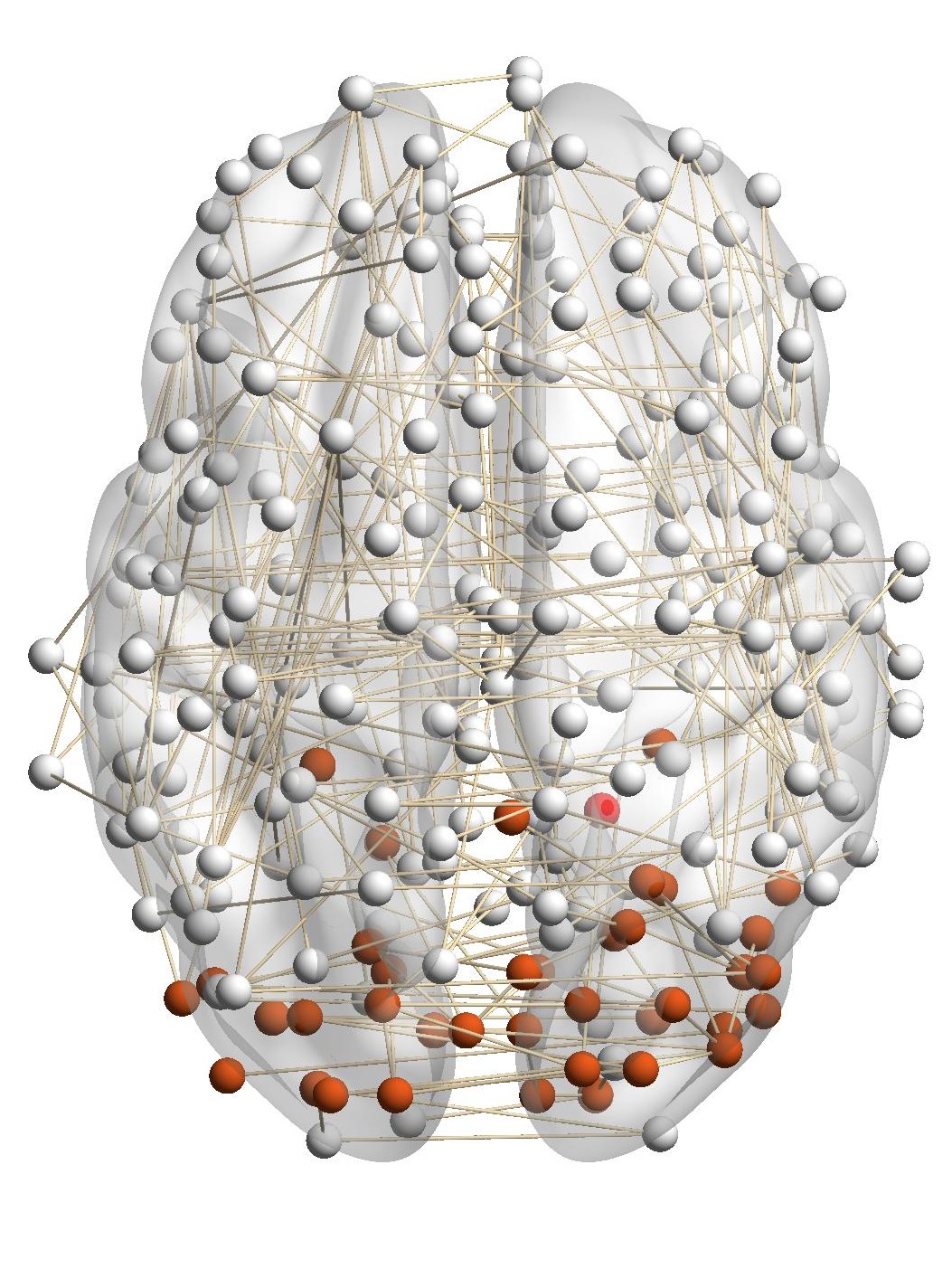}\label{fig:brain_RWRM}}\\ 
\subfigure[RWR]{\includegraphics[width=1.2in]{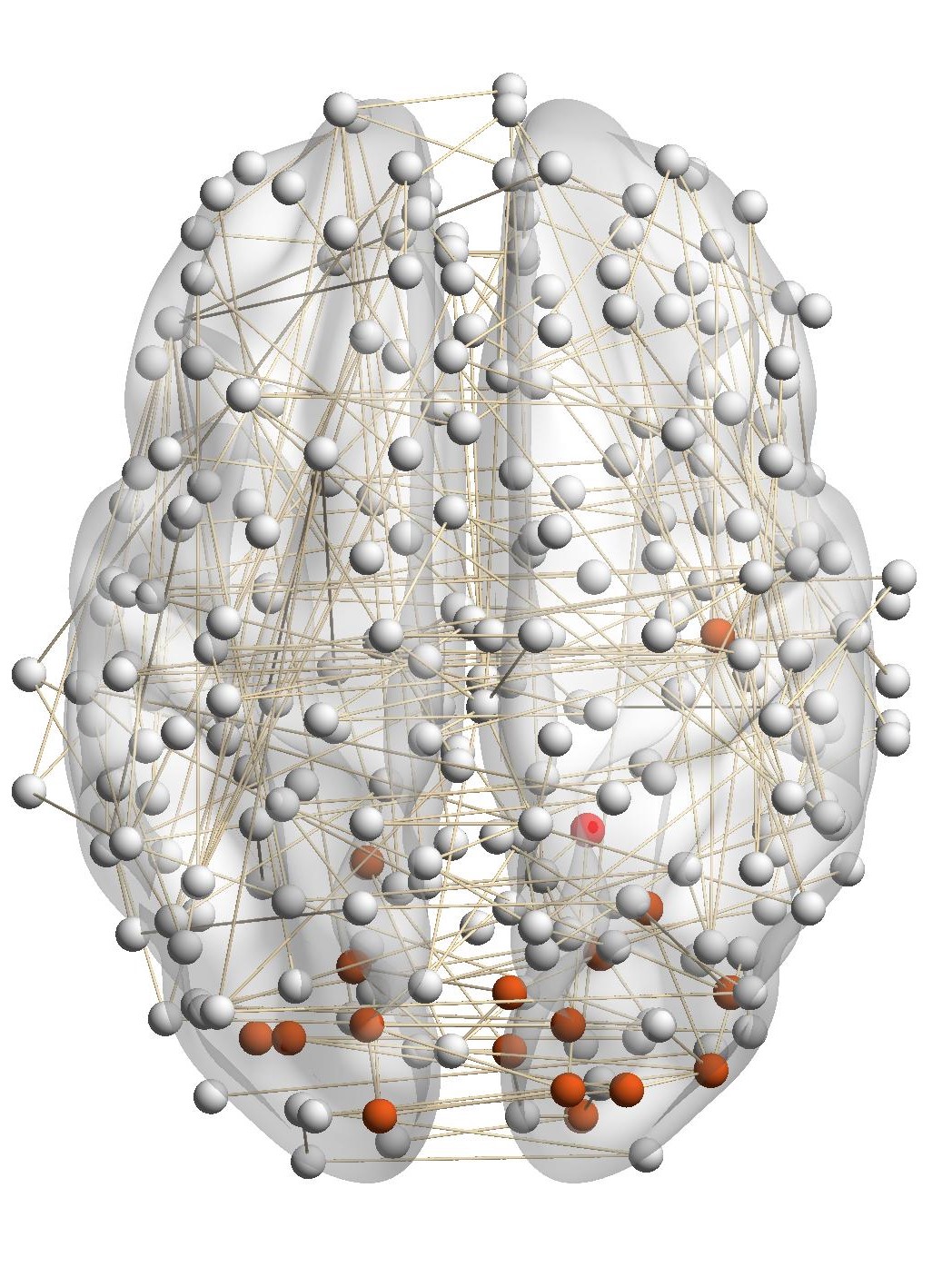}\label{fig:brain_RWR}}
\subfigure[MWC]{\includegraphics[width=1.2in]{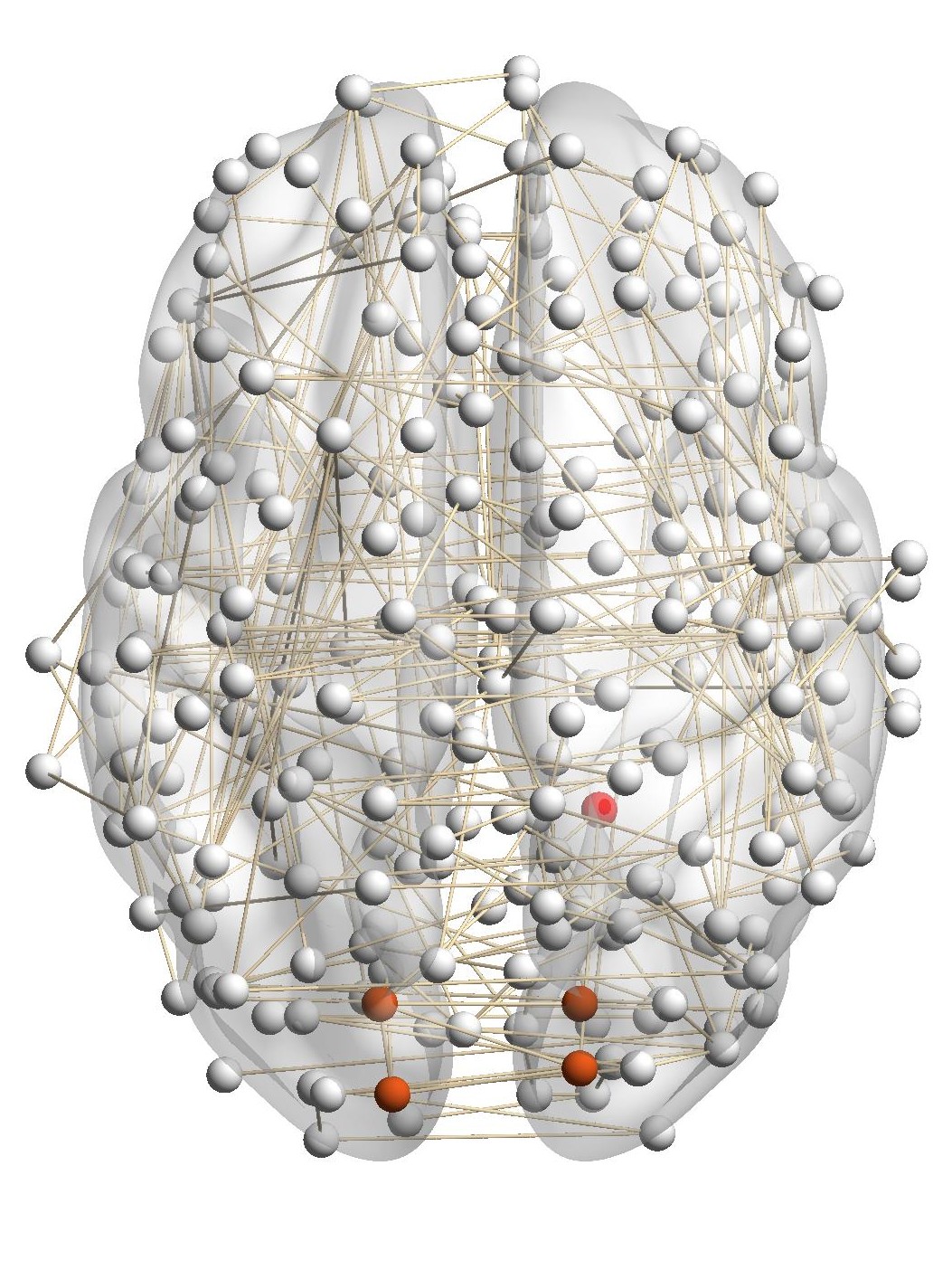}\label{fig:brain_MWC}}
\\
\subfigure[LEMON]{\includegraphics[width=1.2in]{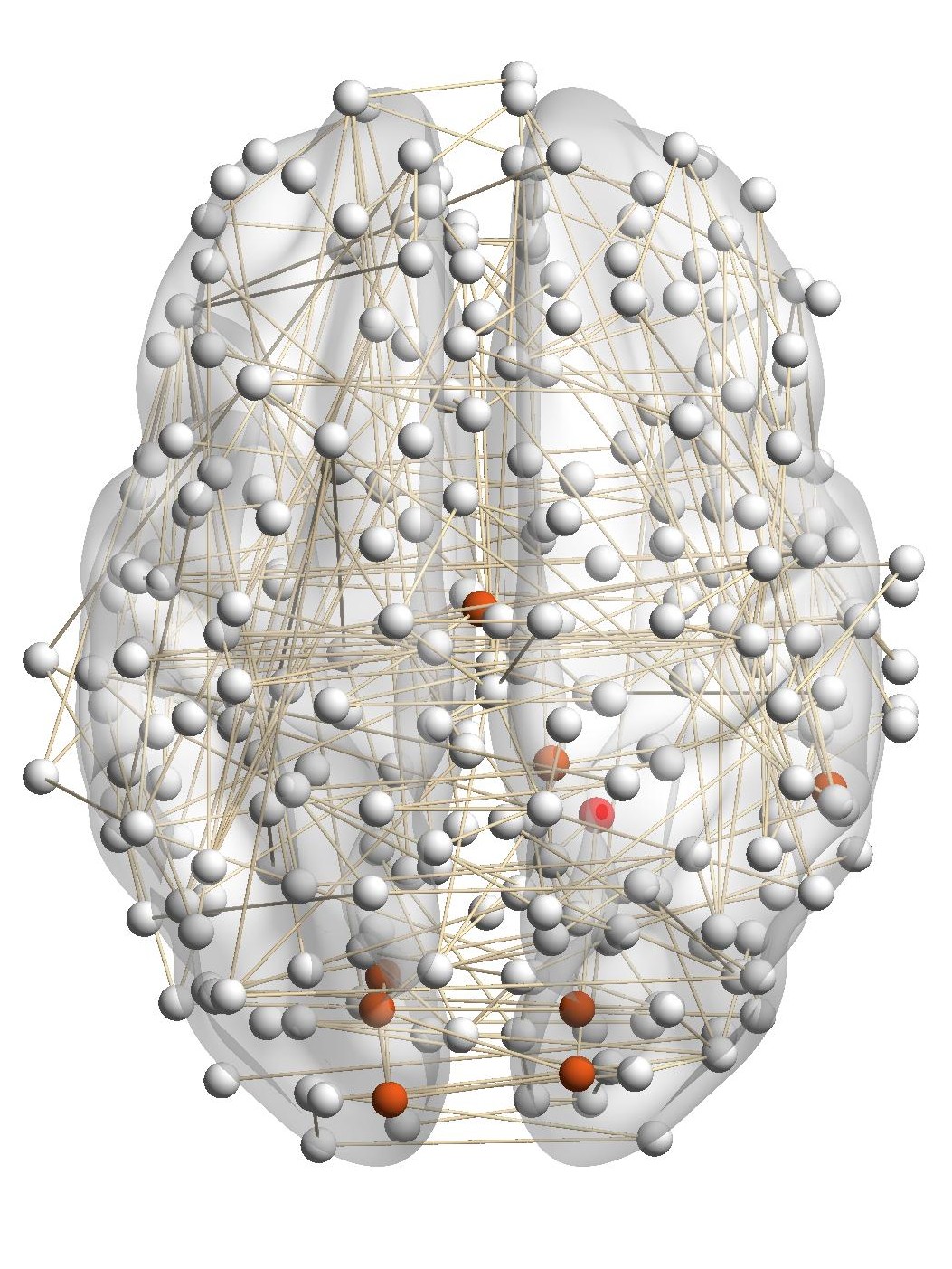}\label{fig:brain_LEMON}}
\subfigure[QDC]{\includegraphics[width=1.2in]{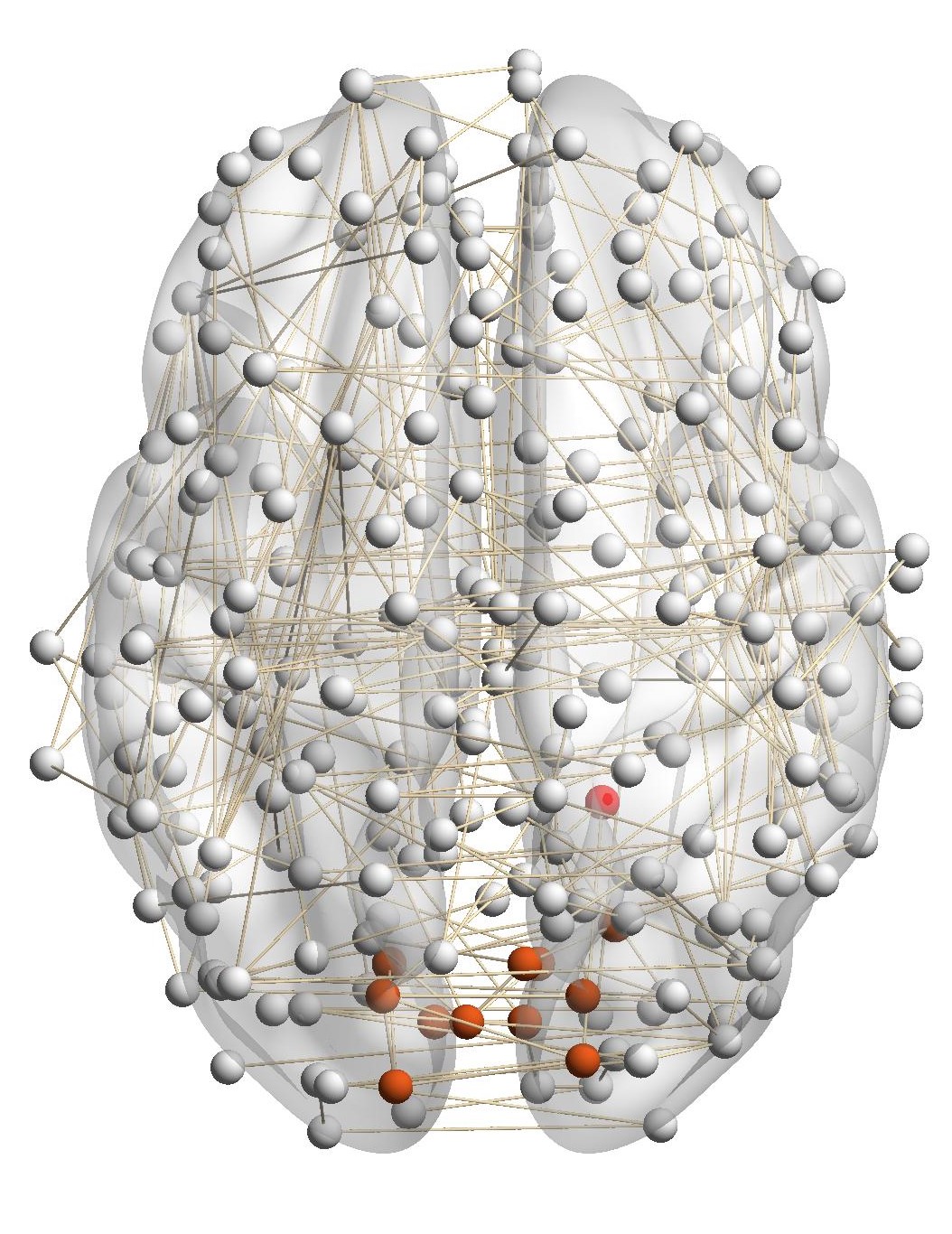}\label{fig:brain_QDC}}\\
\subfigure[MRWR]{\includegraphics[width=1.2in]{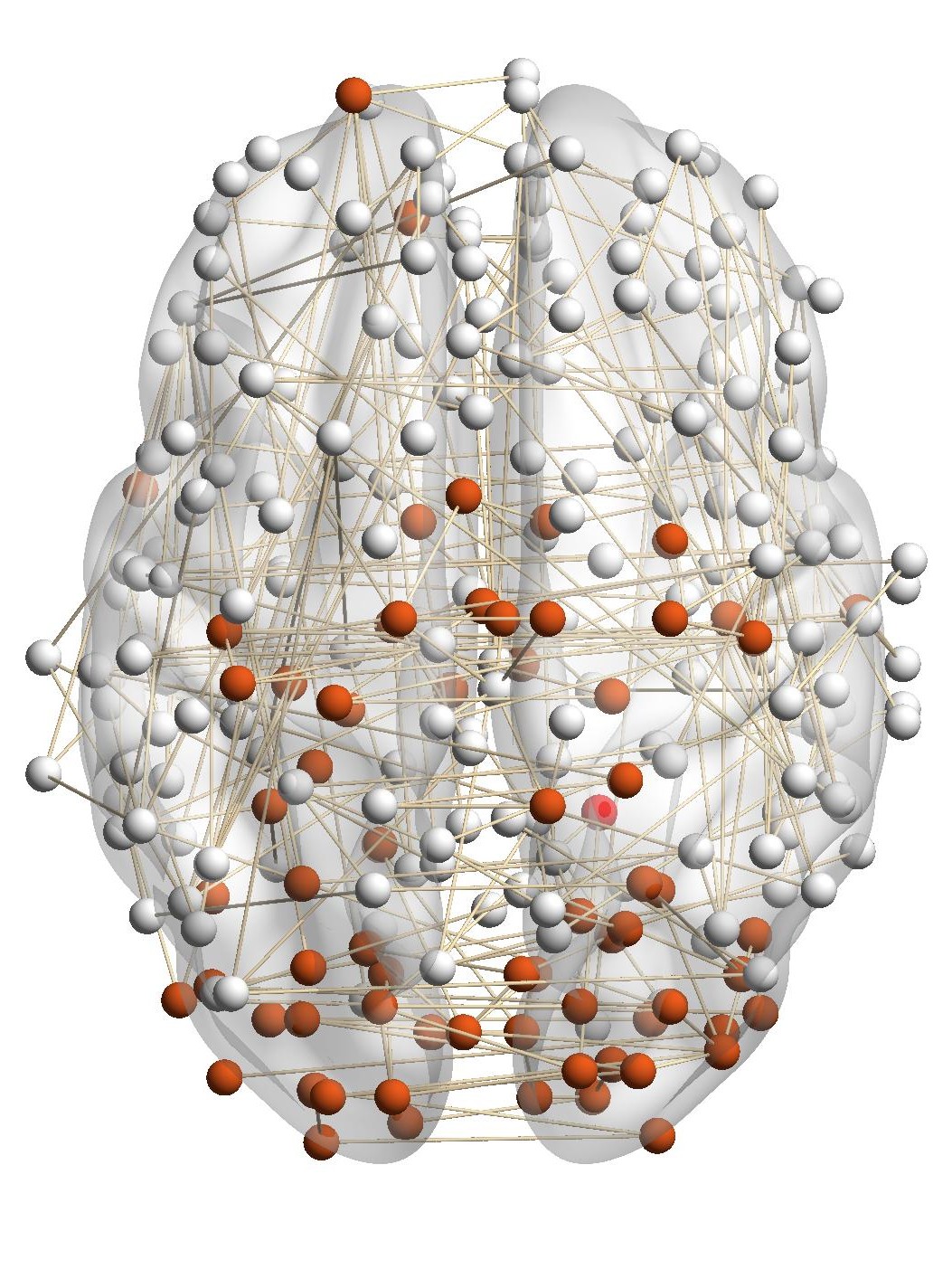}\label{fig:brain_MRWR}}
\subfigure[ML-LCD]{\includegraphics[width=1.2in]{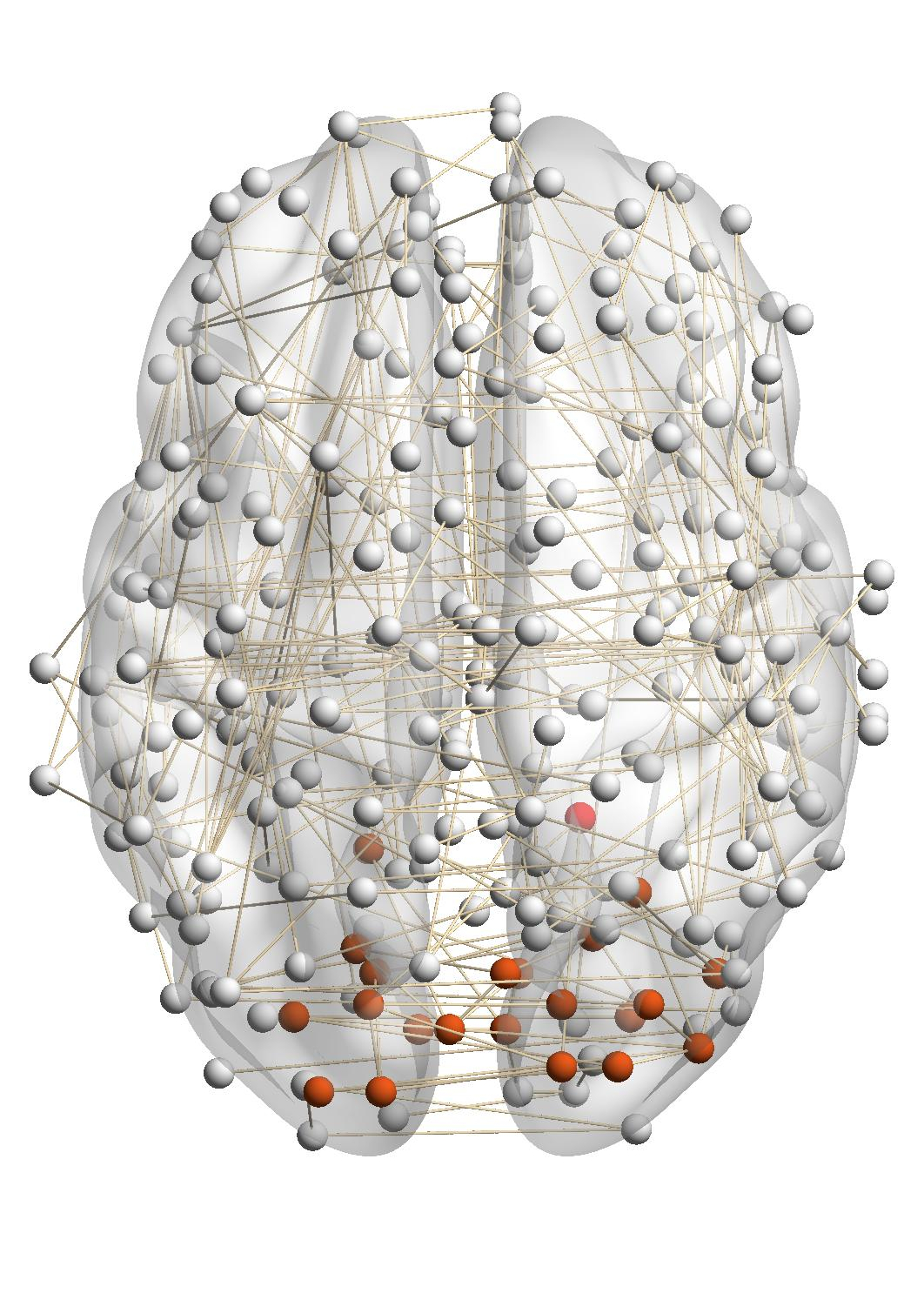}\label{fig:mllcd}}
\caption{Ground truth and detected visual systems by different algorithms given the query node.}
\label{fig:brain_caseStudy}
\end{figure}

\stitle{Case Study on \dblp Dataset.}
In the general multiple networks, for a query node from one network, we only have the ground truth local community in that network but no ground truth about the relevant local communities in other networks. So in this section, we use \dblp as a case study to demonstrate the relevance of local communities detected from other networks by \rwm. We use Prof. Danai Koutra from UMich as the query. The \dblp dataset was collected in May 2014 when she was a Ph.D. student advised by Prof. Christos Faloutsos. Table~\ref{tab:case_dblp} shows the detected author community and paper community. Due to the space limitation, instead of showing the details of the paper community, we list venues where these papers were published in. For example, ``KDD(3)'' means that 3 KDD papers are included in the detected paper community. The table shows that our method can detect local communities from both networks with high qualities. Specifically, the detected authors are mainly from her advisor's group. The detected paper communities are mainly published in Database/Data Mining conferences.

\begin{table}[ht]
    \centering
    \begin{small}
        \caption{Case study on \dblp}
        \label{tab:case_dblp}
    \begin{tabular}{c|c}
        \hline
         Author community&Paper community\\
         \hline
        Danai Koutra (query) & KDD (3) \\
        Christos Faloutsos & PKDD (3)\\
        Evangelos Papalexakis &PAKDD (3)\\
        U Kang & SIGMOD (2)\\
        Tina Eliassi-Rad & ICDM (2)\\
        Michele Berlingerio & TKDE(1)\\
        Duen Horng Chau & ICDE (1)\\
        Leman Akoglu & TKDD (1)\\
        Jure Leskovec & ASONAM (1)\\
Hanghang Tong & WWW(1)\\
        ... &...\\
        \hline
    \end{tabular}
    \end{small}
\end{table}

\subsection{Network Embedding}
Network embedding aims to learn low dimensional representations for nodes in a network to preserve the network structure. These representations are used as features in many tasks on networks, such as node classification~\cite{ZhaoZYXFC18}, recommendation~\cite{wen2018network}, and link prediction~\cite{chen2018pme}. Many network embedding methods consist of two components: one extracts contexts for nodes, and the other learns node representations from the contexts. For the first component, random walk is a routinely used method to generate contexts. For example, DeepWalk~\cite{perozzi2014deepwalk} uses the truncated random walk to sample paths for each node, and node2vec~\cite{grover2016node2vec} uses biased second order random walk. In this part, we apply the proposed \rwm for the multiple network embedding. \mitdata and \ci are used in this part.

\subsubsection{Experimental setup}
To evaluate \rwm on network embedding, we choose two classic network embedding methods: DeepWalk and node2vec, and replace the first components (sampling parts) with \rwm. We denote the DeepWalk and node2vec with replacements by DeepWalk\_RWM and node2vec\_RWM respectively. Specifically, For each node on the target network, \rwm based methods first use the node as the starting node and calculate the static transition matrix of the walker on the target network. Then truncated random walk (DeepWalk\_RWM) and second order random walk (node2vec\_RWM) are applied on the matrix to generate contexts for the node. we compare these two methods with traditional ones. Since we focus on the sampling phase of network embedding, to control the variables, for all four algorithms, we use word2vec (skip-gram model)~\cite{mikolov2013distributed} with the same parameters to generate node embeddings from sampling.

For all methods, the dimensionality of embeddings is 100. The walk length is set to 40 and walks per node is set to 10. For node2vec, we use a grid search over $p, q \in \{0.25, 0.5, 1, 2, 4\}$~\cite{grover2016node2vec}. For \rwm based methods, we set $\lambda = 0.7, \epsilon=0.01, \theta = 0.9$. Windows size of word2vec is set to its default value 5.

\subsubsection{Accuracy Evaluation.}
We compare different methods through a classification task. On each dataset, embedding vectors are learned from the full dataset. Then the embeddings are used as input to the SVM classifier~\cite{scikit-learn}. When training the classifier, we randomly sample a portion of the nodes as the training set and the rest as the testing set. The ratio of training is varied from 10\% to 90\%. We use NMI, Macro-F1, and Micro-F1 scores to evaluate classification accuracy. For \mitdata, we set a network as the target network at one time and report the average results here. For \ci, we only embed the collaboration network, since only authors are labeled. 

Fig. \ref{fig:embedding} shows the experimental results. \rwm based network embedding methods consistently perform better than traditional counterparts in terms of all metrics. This is because \rwm can use additional information on multiple networks to refine node embedding, while baseline methods are affected by noise and incompleteness on a single network. In most cases, node2vec achieves better results than DeepWalk in both the traditional setting and \rwm setting, which is similar to the conclusion drawn in~\cite{grover2016node2vec}. Besides, \rwm based methods are much better than other methods in \ci dataset, when the training ratio is small, \eg 10\%, which means \rwm is more useful in real practice when the available labels are scarce. This advantage comes from more information from other networks when embedding the target network. 

\begin{figure}[ht!]
\centering
\subfigure[NMI in \mitdata]{\includegraphics[width=1.6in]{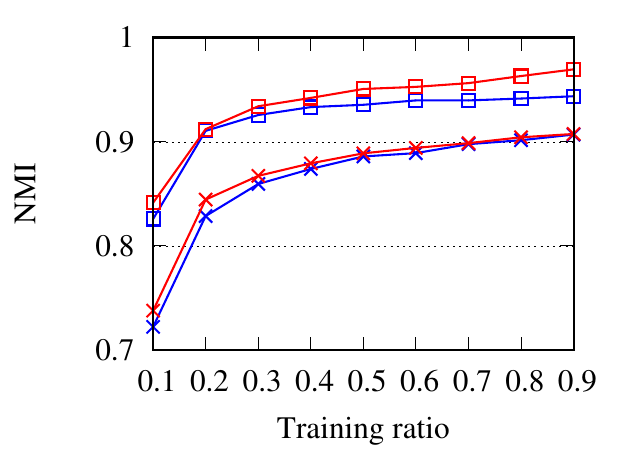}\label{fig:embeddingRMACC}}
\subfigure[NMI in \ci]{\includegraphics[width=1.6in]{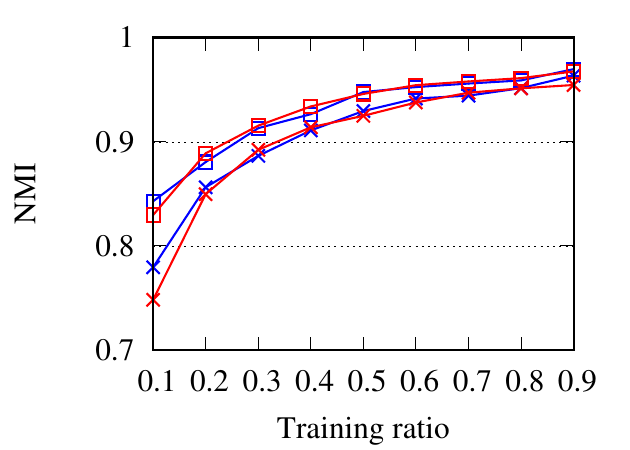}\label{fig:embeddingCSACC}}
\subfigure[Macro-F1 in \mitdata]{\includegraphics[width=1.6in]{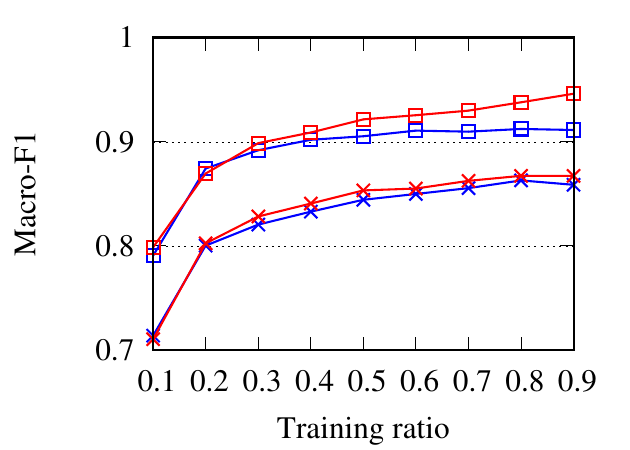}\label{fig:embeddingRMMA}}
\subfigure[Macro-F1 in \ci]{\includegraphics[width=1.6in]{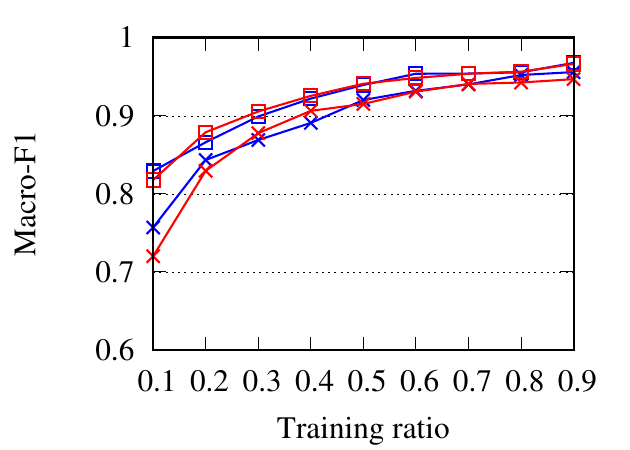}\label{fig:embeddingCSMA}}
\subfigure[Micro-F1  in \mitdata]{\includegraphics[width=1.6in]{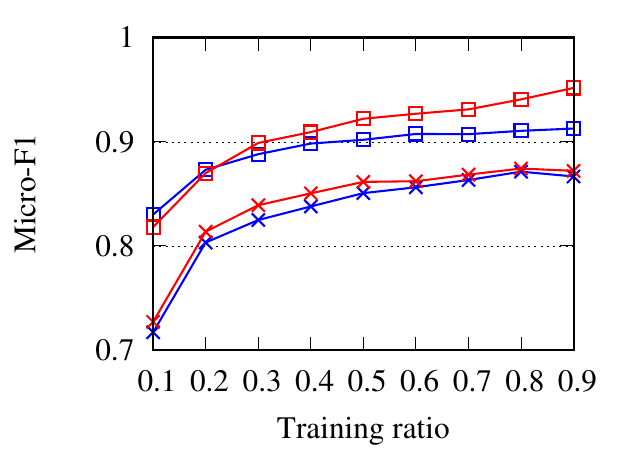}\label{fig:embeddingRMMI}}
\subfigure[Micro-F1 in  \ci]{\includegraphics[width=1.6in]{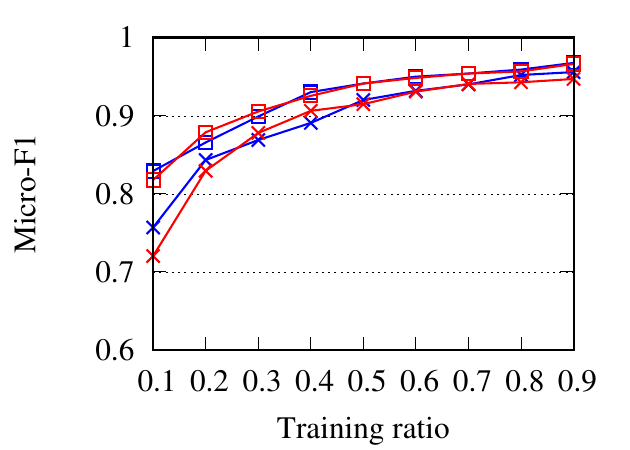}\label{fig:embeddingCSMI}}
\subfigure{\centering \includegraphics[width=3.5in]{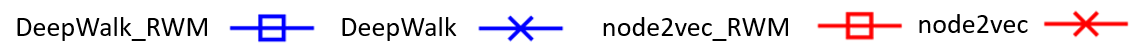}}
\caption{Accuracy performances of network embedding}
\label{fig:embedding}
\end{figure}

\subsection{Link Prediction}
We chose two multiplex networks, \gene and \athle to test our approach.

\subsubsection{Experimental setup}
Here we apply \rwm based random walk with restart (RWR) on multiplex networks to calculate the proximity scores between each pair of nodes on the target network, denoted by RWM. Then we choose the top 100 pairs of unconnected nodes with the highest proximity scores as the predicted links. We compare this method with traditional RWR on the target networks, denoted by Single, and RWR on the merged networks, which is obtained by summing up the weights of the same edges in all networks, denoted by Merged.  We choose the probability of restart $\alpha=0.9$ for all three methods and decay factor $\lambda = 0.4$ for \rwm. 

We focus on one network, denoted by the target network, at one time. For the target network, we randomly remove 30\% edges as the probe edges, which are used as the ground truth for testing. Other networks are unchanged~\cite{lu2009similarity}. We use ${precision}_{100}$ as the evaluation measurement. The final results are averaged over all networks in the dataset.

\subsubsection{Experimental Results}
As shown in Fig.~\ref{fig:lp}, for \athle, our method consistently performs better than the other two methods. Specifically, \rwm based RWR outperforms the traditional RWR and RWR on the merged network by (93.37\%, 218\%) in $G_1$, (351.69\%, 61.3\%) in $G_2$, (422.31\%, 8.17\%) in $G_3$, in average, respectively. For \gene, \rwm based RWR achieves the best results in 6 out of 7 networks. From the experimental results, we can draw a conclusion that using auxiliary link information from other networks can improve link prediction accuracy. Single outperforms Merged in 1 out of 3 and 4 out of 7 networks in these two multiplex networks, respectively, which shows that including other networks does not always lead to better results. Thus, \rwm is an effective way to actively select useful information from other networks.

\begin{figure}[ht]
\subfigure[\athle]{\includegraphics[width=1.65in]{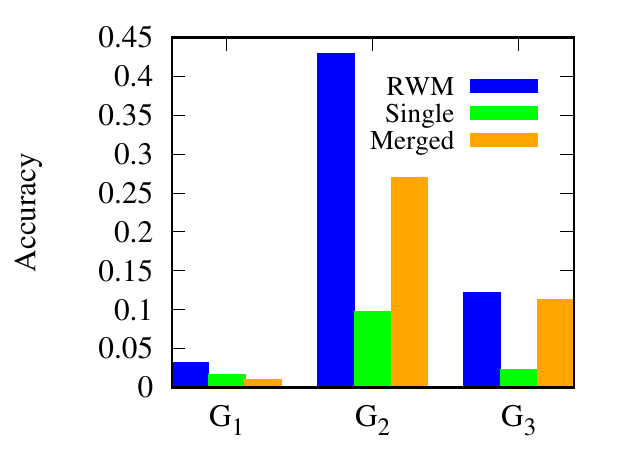}\label{fig:lpsocial}}
\subfigure[\gene]{\includegraphics[width=1.65in]{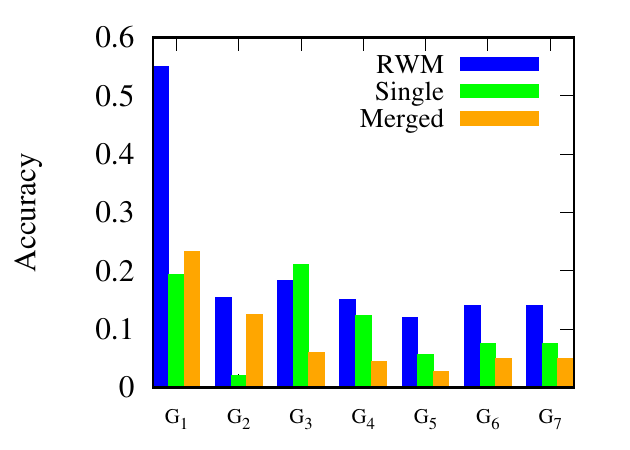}\label{fig:lpgenetic}}
\caption{Accuracy Performances of Link Prediction}
\label{fig:lp}
\end{figure}

\subsection{Efficiency Evaluation}
Note that the running time of \rwm using basic power-iteration (Algorithm  \ref{alg:two_multiplex}) is similar to the iteration-based random walk local community detection methods but \rwm obtains better performance than other baselines with a large margin (Table \ref{tab:ac}). Thus, in this section, we focus on \rwm and use synthetic datasets to evaluate its efficiency. There are three methods to update visiting probabilities in \rwm: (1) power iteration method in \textcolor{black}{Algorithm  \ref{alg:two_multiplex}} (PowerIte), (2) power iteration with early stopping introduced in Sec. \ref{A1} (A1), and (3) partial updating in Sec. \ref{A2} (A2). 

\begin{figure}[h]
\centering
\subfigure[{{Total time \wrt \# of networks}}]{\includegraphics[width=1.6in]{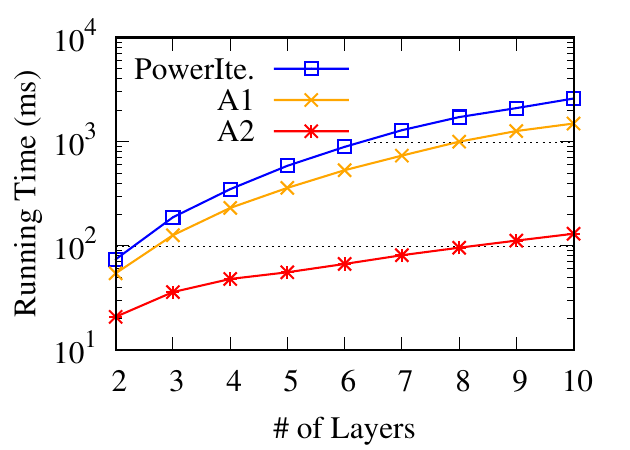}\label{fig:time_network_total}}\hspace{-0.4em}
\subfigure[Total time \wrt \# of nodes]{\includegraphics[width=1.6in]{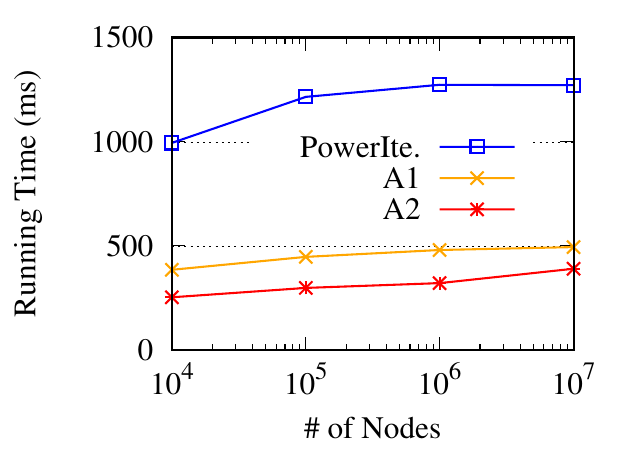}\label{fig:time_node_total}}\\
\caption{Efficiency evaluation with synthetic datasets}
\label{fig:efficiency}
\end{figure}

We first evaluate the running time \wrt the number of networks. We generate 9 datasets with different numbers of networks (2 to 10). For each dataset,  we first use a graph generator~\cite{lancichinetti2009benchmarks} to generate a base network consisting of 1,000 nodes and about 7,000 edges. Then we obtain multiple networks from the base network. In each network, we randomly delete 50\% edges from the base network. For each dataset, we randomly select a node as the query and detect local communities. In Fig.~\ref{fig:time_network_total}, we report the running time of the three methods averaged over 100 trials.
The early stopping in A1 saves time by about 2 times for the iteration method. The partial updating in A2 can further speed up by about 20 times. Furthermore, we can observe that the running time of A2 grows slower than PowerIte. Thus the efficiency gain tends to be larger with more networks.

Next, we evaluate the running time \wrt the number of nodes. Similar to the aforementioned generating process, we synthesize 4 datasets from 4 base networks with the number of nodes ranging from $10^4$ to $10^7$. The average node degree is 7 in each base network. For each dataset, we generate multiplex networks with three networks, by randomly removing half of the edges from the base network. The running time is shown in Fig.~\ref{fig:time_node_total}. Compared to the power iteration method, the approximation methods are much faster, especially when the number of nodes is large.  In addition, they grow much slower, which enables their scalability on even larger datasets.

We further check the number of visited nodes, which have positive visiting probability scores in \rwm, in different-sized networks. In Table \ref{tab:visit}, we show the number of visited nodes both at the splitting time $T_e$ and the end of updating (averaged over 100 trials). Note that we compute $T_e=\lceil\log_\lambda \frac{\epsilon(1-\lambda)}{K}\rceil$ according to Theorem \ref{th:error_general} with $K=3, \epsilon=0.01, \lambda=0.7$. We see that in the end, only a very small proportion of nodes are visited (for the biggest $10^7$ network, probability only propagates to 0.02\% nodes). This demonstrates the locality of \rwm. Besides, in the first phase (early stop at $T_e$), visiting probabilities are restricted to about 50 nodes around the query node.

\begin{table}[t!]
\begin{small}
    \centering
        \caption{Number of visited nodes v.s. different network size}
        \label{tab:visit}
    \begin{tabular}{c|c|c|c|c}
    \hline
         & $10^4$ &$10^5$ & $10^6$ & $10^7$ \\
        \hline
        \# visited nodes ($T_e$) & 47.07	& 50.58	& 53.27	& 53.54\\
        \# visited nodes (end)  &1977.64&2103.43&2125.54& 2177.70\\\hline
    \end{tabular}
\end{small}
\end{table}
\section{Conclusion}
In this paper, we propose a novel random walk model, \rwm, on multiple networks. Random walkers on different networks sent by \rwm mutually affect their visiting probabilities in a reinforced manner. By aggregating their effects from local subgraphs in different networks, \rwm restricts walkers' most visiting probabilities among the most relevant nodes. Rigorous theoretical foundations are provided to verify the effectiveness of \rwm. Two speeding-up strategies are also developed for efficient computation. Extensive experimental results verify the advantages of \rwm in effectiveness and efficiency.

\section*{ACKNOWLEDGMENTS}
This project was partially supported by NSF project IIS-1707548.

\bibliographystyle{abbrv}
\begin{small}
\bibliography{references}

\begin{thebibliography}{10}

\bibitem{alamgir2010multi}
M.~Alamgir and U.~Von~Luxburg.
\newblock Multi-agent random walks for local clustering on graphs.
\newblock In {\em ICDM}, pages 18--27, 2010.

\bibitem{andersen2006local}
R.~Andersen, F.~Chung, and K.~Lang.
\newblock Local graph partitioning using pagerank vectors.
\newblock In {\em FOCS}, pages 475--486, 2006.

\bibitem{kCoreEfficient}
N.~Barbieri, F.~Bonchi, E.~Galimberti, and F.~Gullo.
\newblock Efficient and effective community search.
\newblock {\em Data mining and knowledge discovery}, 29(5):1406--1433, 2015.

\bibitem{bian2019memory}
Y.~Bian, D.~Luo, Y.~Yan, W.~Cheng, W.~Wang, and X.~Zhang.
\newblock Memory-based random walk for multi-query local community detection.
\newblock {\em Knowledge and Information Systems}, pages 1--35, 2019.

\bibitem{bian2017many}
Y.~Bian, J.~Ni, W.~Cheng, and X.~Zhang.
\newblock Many heads are better than one: Local community detection by the
  multi-walker chain.
\newblock In {\em ICDM}, pages 21--30, 2017.

\bibitem{MWC_2}
Y.~Bian, J.~Ni, W.~Cheng, and X.~Zhang.
\newblock The multi-walker chain and its application in local community
  detection.
\newblock {\em Knowledge and Information Systems}, 60(3):1663--1691, 2019.

\bibitem{bian2018multi}
Y.~Bian, Y.~Yan, W.~Cheng, W.~Wang, D.~Luo, and X.~Zhang.
\newblock On multi-query local community detection.
\newblock In {\em ICDM}, pages 9--18, 2018.

\bibitem{bojchevski2018netgan}
A.~Bojchevski, O.~Shchur, D.~Z{\"u}gner, and S.~G{\"u}nnemann.
\newblock Netgan: Generating graphs via random walks.
\newblock In {\em International conference on machine learning}, pages
  610--619. PMLR, 2018.

\bibitem{cao2015grarep}
S.~Cao, W.~Lu, and Q.~Xu.
\newblock Grarep: Learning graph representations with global structural
  information.
\newblock In {\em CIKM}, 2015.

\bibitem{chen2018pme}
H.~Chen, H.~Yin, W.~Wang, H.~Wang, Q.~V.~H. Nguyen, and X.~Li.
\newblock Pme: projected metric embedding on heterogeneous networks for link
  prediction.
\newblock In {\em SIGKDD}, 2018.

\bibitem{cheng2013flexible}
W.~Cheng, X.~Zhang, Z.~Guo, Y.~Wu, P.~F. Sullivan, and W.~Wang.
\newblock Flexible and robust co-regularized multi-domain graph clustering.
\newblock In {\em SIGKDD}, pages 320--328. ACM, 2013.

\bibitem{de2015identifying}
M.~De~Domenico, A.~Lancichinetti, A.~Arenas, and M.~Rosvall.
\newblock Identifying modular flows on multilayer networks reveals highly
  overlapping organization in interconnected systems.
\newblock {\em Physical Review X}, 5(1):011027, 2015.

\bibitem{de2015structural}
M.~De~Domenico, V.~Nicosia, A.~Arenas, and V.~Latora.
\newblock Structural reducibility of multilayer networks.
\newblock {\em Nature communications}, 6:6864, 2015.

\bibitem{masha2018multi}
M.~Ghorbani, M.~S. Baghshah, and H.~R. Rabiee.
\newblock Multi-layered graph embedding with graph convolutional networks.
\newblock {\em arXiv preprint arXiv:1811.08800}, 2018.

\bibitem{grover2016node2vec}
A.~Grover and J.~Leskovec.
\newblock node2vec: Scalable feature learning for networks.
\newblock In {\em SIGKDD}, pages 855--864, 2016.

\bibitem{interdonato2017local}
R.~Interdonato, A.~Tagarelli, D.~Ienco, A.~Sallaberry, and P.~Poncelet.
\newblock Local community detection in multilayer networks.
\newblock {\em Data Mining and Knowledge Discovery}, 31(5):1444--1479, 2017.

\bibitem{jeh2002simrank}
G.~Jeh and J.~Widom.
\newblock Simrank: a measure of structural-context similarity.
\newblock In {\em KDD}, pages 538--543, 2002.

\bibitem{kim2015community}
J.~Kim and J.-G. Lee.
\newblock Community detection in multi-layer graphs: A survey.
\newblock {\em ACM SIGMOD Record}, 44(3):37--48, 2015.

\bibitem{MultilayerNetworks}
M.~Kivel{\"a}, A.~Arenas, M.~Barthelemy, J.~P. Gleeson, Y.~Moreno, and M.~A.
  Porter.
\newblock Multilayer networks.
\newblock {\em Journal of complex networks}, 2(3):203--271, 2014.

\bibitem{klicpera2018predict}
J.~Klicpera, A.~Bojchevski, and S.~G{\"u}nnemann.
\newblock Predict then propagate: Graph neural networks meet personalized
  pagerank.
\newblock In {\em ICLR}, 2018.

\bibitem{klicpera2019diffusion}
J.~Klicpera, S.~Wei{\ss}enberger, and S.~G{\"u}nnemann.
\newblock Diffusion improves graph learning.
\newblock In {\em NeurIPS}, 2019.

\bibitem{kloster2014heat}
K.~Kloster and D.~F. Gleich.
\newblock Heat kernel based community detection.
\newblock In {\em SIGKDD}, pages 1386--1395, 2014.

\bibitem{lancichinetti2009benchmarks}
A.~Lancichinetti and S.~Fortunato.
\newblock Benchmarks for testing community detection algorithms on directed and
  weighted graphs with overlapping communities.
\newblock {\em Physical Review E}, 2009.

\bibitem{langville2011google}
A.~N. Langville and C.~D. Meyer.
\newblock {\em Google's PageRank and beyond: The science of search engine
  rankings}.
\newblock Princeton University Press, 2011.

\bibitem{leskovec2012learning}
J.~Leskovec and J.~J. Mcauley.
\newblock Learning to discover social circles in ego networks.
\newblock In {\em NIPS}, pages 539--547, 2012.

\bibitem{li2018mane}
J.~Li, C.~Chen, H.~Tong, and H.~Liu.
\newblock Multi-layered network embedding.
\newblock In {\em SDM}, 2018.

\bibitem{li2015uncovering}
Y.~Li, K.~He, D.~Bindel, and J.~E. Hopcroft.
\newblock Uncovering the small community structure in large networks: A local
  spectral approach.
\newblock In {\em WWW}, pages 658--668, 2015.

\bibitem{lim2016bibliographic}
K.~W. Lim and W.~Buntine.
\newblock Bibliographic analysis on research publications using authors,
  categorical labels and the citation network.
\newblock {\em Machine Learning}, 103(2):185--213, 2016.

\bibitem{liu2015robust}
R.~Liu, W.~Cheng, H.~Tong, W.~Wang, and X.~Zhang.
\newblock Robust multi-network clustering via joint cross-domain cluster
  alignment.
\newblock In {\em ICDM}, pages 291--300. IEEE, 2015.

\bibitem{lovasz1993random}
L.~Lov{\'a}sz et~al.
\newblock Random walks on graphs: A survey.
\newblock {\em Combinatorics, Paul erdos is eighty}, 2(1):1--46, 1993.

\bibitem{LZ:StoPro}
D.~Lu and H.~Zhang.
\newblock {\em Stochastic process and applications}.
\newblock Tsinghua University Press, 1986.

\bibitem{lu2009similarity}
L.~L{\"u}, C.-H. Jin, and T.~Zhou.
\newblock Similarity index based on local paths for link prediction of complex
  networks.
\newblock {\em Physical Review E}, 80(4):046122, 2009.

\bibitem{luo2020deep}
D.~Luo, J.~Ni, S.~Wang, Y.~Bian, X.~Yu, and X.~Zhang.
\newblock Deep multi-graph clustering via attentive cross-graph association.
\newblock In {\em WSDM}, pages 393--401, 2020.

\bibitem{shuai2018query}
S.~Ma, C.~Gong, R.~Hu, D.~Luo, C.~Hu, and J.~Huai.
\newblock Query independent scholarly article ranking.
\newblock In {\em ICDE}, 2018.

\bibitem{martinez2017survey}
V.~Mart{\'\i}nez, F.~Berzal, and J.-C. Cubero.
\newblock A survey of link prediction in complex networks.
\newblock {\em ACM Computing Surveys}, 2017.

\bibitem{mikolov2013distributed}
T.~Mikolov, I.~Sutskever, K.~Chen, G.~S. Corrado, and J.~Dean.
\newblock Distributed representations of words and phrases and their
  compositionality.
\newblock In {\em NIPS}, 2013.

\bibitem{MultiplexNet}
P.~J. Mucha, T.~Richardson, K.~Macon, M.~A. Porter, and J.-P. Onnela.
\newblock Community structure in time-dependent, multiscale, and multiplex
  networks.
\newblock {\em Science}, 328(5980):876--878, 2010.

\bibitem{mucha2010community}
P.~J. Mucha, T.~Richardson, K.~Macon, M.~A. Porter, and J.-P. Onnela.
\newblock Community structure in time-dependent, multiscale, and multiplex
  networks.
\newblock {\em science}, 328(5980):876--878, 2010.

\bibitem{newman2004coauthorship}
M.~E. Newman.
\newblock Coauthorship networks and patterns of scientific collaboration.
\newblock {\em Proceedings of the national academy of sciences}, 101(suppl
  1):5200--5205, 2004.

\bibitem{ni2018co}
J.~Ni, S.~Chang, X.~Liu, W.~Cheng, H.~Chen, D.~Xu, and X.~Zhang.
\newblock Co-regularized deep multi-network embedding.
\newblock In {\em WWW}, pages 469--478, 2018.

\bibitem{ni2018comclus}
J.~Ni, W.~Cheng, W.~Fan, and X.~Zhang.
\newblock Comclus: A self-grouping framework for multi-network clustering.
\newblock {\em IEEE Transactions on Knowledge and Data Engineering},
  30(3):435--448, 2018.

\bibitem{omodei2015characterizing}
E.~Omodei, M.~D. De~Domenico, and A.~Arenas.
\newblock Characterizing interactions in online social networks during
  exceptional events.
\newblock {\em Frontiers in Physics}, 3:59, 2015.

\bibitem{ou2017multi}
L.~Ou-Yang, H.~Yan, and X.-F. Zhang.
\newblock A multi-network clustering method for detecting protein complexes
  from multiple heterogeneous networks.
\newblock {\em BMC bioinformatics}, 2017.

\bibitem{page1999pagerank}
L.~Page, S.~Brin, R.~Motwani, and T.~Winograd.
\newblock The pagerank citation ranking: Bringing order to the web.
\newblock Technical report, Stanford InfoLab, 1999.

\bibitem{scikit-learn}
F.~Pedregosa, G.~Varoquaux, A.~Gramfort, V.~Michel, B.~Thirion, O.~Grisel,
  M.~Blondel, P.~Prettenhofer, R.~Weiss, V.~Dubourg, J.~Vanderplas, A.~Passos,
  D.~Cournapeau, M.~Brucher, M.~Perrot, and E.~Duchesnay.
\newblock Scikit-learn: Machine learning in {P}ython.
\newblock {\em JMLR}, 2011.

\bibitem{perozzi2014deepwalk}
B.~Perozzi, R.~Al-Rfou, and S.~Skiena.
\newblock Deepwalk: Online learning of social representations.
\newblock In {\em SIGKDD}, 2014.

\bibitem{tang2015line}
J.~Tang, M.~Qu, M.~Wang, M.~Zhang, J.~Yan, and Q.~Mei.
\newblock Line: Large-scale information network embedding.
\newblock In {\em WWW}, pages 1067--1077, 2015.

\bibitem{tang2008arnetminer}
J.~Tang, J.~Zhang, L.~Yao, J.~Li, L.~Zhang, and Z.~Su.
\newblock Arnetminer: extraction and mining of academic social networks.
\newblock In {\em SIGKDD}, pages 990--998, 2008.

\bibitem{tong2006fast}
H.~Tong, C.~Faloutsos, and J.-Y. Pan.
\newblock Fast random walk with restart and its applications.
\newblock In {\em ICDM}, pages 613--622, 2006.

\bibitem{van2013wu}
D.~C. Van~Essen, S.~M. Smith, D.~M. Barch, T.~E. Behrens, E.~Yacoub,
  K.~Ugurbil, W.-M.~H. Consortium, et~al.
\newblock The wu-minn human connectome project: an overview.
\newblock {\em Neuroimage}, 80:62--79, 2013.

\bibitem{veldt2019flow}
N.~Veldt, C.~Klymko, and D.~F. Gleich.
\newblock Flow-based local graph clustering with better seed set inclusion.
\newblock In {\em SDM}, pages 378--386, 2019.

\bibitem{wen2018network}
Y.~Wen, L.~Guo, Z.~Chen, and J.~Ma.
\newblock Network embedding based recommendation method in social networks.
\newblock In {\em WWW}, 2018.

\bibitem{wu2016remember}
Y.~Wu, Y.~Bian, and X.~Zhang.
\newblock Remember where you came from: on the second-order random walk based
  proximity measures.
\newblock {\em Proceedings of the VLDB Endowment}, 10(1):13--24, 2016.

\bibitem{wu2015robust}
Y.~Wu, R.~Jin, J.~Li, and X.~Zhang.
\newblock Robust local community detection: on free rider effect and its
  elimination.
\newblock {\em PVLDB}, 8(7):798--809, 2015.

\bibitem{wu2017remember}
Y.~Wu, X.~Zhang, Y.~Bian, Z.~Cai, X.~Lian, X.~Liao, and F.~Zhao.
\newblock Second-order random walk-based proximity measures in graph analysis:
  formulations and algorithms.
\newblock {\em The VLDB Journal}, 27(1):127--152, 2018.

\bibitem{xia2021graph}
L.~Xia, Y.~Xu, C.~Huang, P.~Dai, and L.~Bo.
\newblock Graph meta network for multi-behavior recommendation.
\newblock In {\em Proceedings of the 44th International ACM SIGIR Conference on
  Research and Development in Information Retrieval}, pages 757--766, 2021.

\bibitem{yan2019constrained}
Y.~Yan, Y.~Bian, D.~Luo, D.~Lee, and X.~Zhang.
\newblock Constrained local graph clustering by colored random walk.
\newblock In {\em WWW}, pages 2137--2146, 2019.

\bibitem{yao2018local}
Y.~Yan, D.~Luo, J.~Ni, H.~Fei, W.~Fan, X.~Yu, J.~Yen, and X.~Zhang.
\newblock Local graph clustering by multi-network random walk with restart.
\newblock In {\em PAKDD}, pages 490--501, 2018.

\bibitem{yu2022multiplex}
P.~Yu, C.~Fu, Y.~Yu, C.~Huang, Z.~Zhao, and J.~Dong.
\newblock Multiplex heterogeneous graph convolutional network.
\newblock In {\em Proceedings of the 28th ACM SIGKDD Conference on Knowledge
  Discovery and Data Mining}, pages 2377--2387, 2022.

\bibitem{zhang2022role}
H.~Zhang and G.~Kou.
\newblock Role-based multiplex network embedding.
\newblock In {\em International Conference on Machine Learning}, pages
  26265--26280. PMLR, 2022.

\bibitem{zhang2020multiplex}
W.~Zhang, J.~Mao, Y.~Cao, and C.~Xu.
\newblock Multiplex graph neural networks for multi-behavior recommendation.
\newblock In {\em Proceedings of the 29th ACM International Conference on
  Information \& Knowledge Management}, pages 2313--2316, 2020.

\bibitem{ZhaoZYXFC18}
W.~Zhao, J.~Zhu, M.~Yang, D.~Xiao, G.~P.~C. Fung, and X.~Chen.
\newblock A semi-supervised network embedding model for protein complexes
  detection.
\newblock In {\em AAAI}, 2018.

\end{thebibliography}
\end{small}

\end{document}